%% file: journal.tex
\documentclass{birkjour}
\input{preamble.tex} 
\begin{document}

%
%
%
%
%
%
%
%
%

\title
 {Sequent systems for negative modalities}

\author{Ori Lahav}
\address{Max Planck Institute for Software Systems (MPI-SWS), Germany}
\email{orilahav@mpi-sws.org}

\author{Jo\~ao Marcos}
\address{Federal University of Rio Grande do Norte, Brazil}
\email{jmarcos@dimap.ufrn.br}

\author{Yoni Zohar}
\address{Tel Aviv University, Israel}
\email{yoni.zohar@cs.tau.ac.il}

\keywords{Negative modalities, sequent systems, cut-admissibility, analyticity.}


\begin{abstract}
Non-classical negations may fail to be contradictory-forming operators in more than one way, and they often fail also to respect fundamental meta-logical properties such as the replacement property.  Such drawbacks are witnessed by intricate semantics and proof systems, whose philosophical interpretations and computational properties are found wanting.
In this paper we investigate congruential non-classical negations that live inside very natural systems of normal modal logics over complete distributive lattices; these logics are further enriched by adjustment connectives that  
may be used for handling reasoning under uncertainty caused by inconsistency or undeterminedness.
Using such straightforward semantics, we study the classes of frames characterized by seriality, reflexivity, functionality, symmetry, transitivity, and some combinations thereof, and discuss what they reveal about sub-classical properties of negation. 
To the logics thereby characterized we apply a general mechanism that allows one to endow them with analytic ordinary sequent systems, most of which are even cut-free.  
We also investigate the exact circumstances that allow for classical negation to be explicitly defined inside our logics.%
\footnote[1]{A preliminary and abbreviated version of the results in this paper was presented at the \textit{$11^{\mathrm{th}}\!$ International Conference on Advances in Modal Logic} (cf.\ \cite{lahav_it_2016}).}
\end{abstract}

\maketitle

\section{Capturing the impossible, and its dual}
\label{sec:goals}
 \input{sec-goals}

\section{On negative modalities}
\label{sec:intro}
 \input{sec-intro}

\section{A proof system for $PK$}
\label{sec:proofsystem}
 \input{sec-proofsystem}

\section{(Almost) Free Lunch: cut-admissibility and analyticity}
\label{sec:analyticity}
 \input{sec-analyticity}

\section{Some special classes of frames}
\label{sec:extensions}
 \input{sec-extensions}

\section{Definability of classical negation}
\label{sec:definability}
 \input{sec-definability}



\section{This is possibly not the end}
\label{sec:coda}
 \input{sec-coda}


\bibliographystyle{plain}
\bibliography{journal}

%

\end{document}

%% file: preamble.tex
\usepackage{amsthm}
\usepackage{cleveref}

\newtheorem{theorem}{Theorem}[section]
 \newtheorem{corollary}[theorem]{Corollary}
 \newtheorem{lemma}[theorem]{Lemma}
 
 \theoremstyle{definition}
 \theoremstyle{remark}

 \numberwithin{equation}{section}

\makeatletter
\DeclareRobustCommand*\cal{\@fontswitch\relax\mathcal}
\makeatother

\usepackage{graphicx}
\usepackage{tikz}
\usepackage{amsmath,amssymb}
\usepackage{xcolor}
\usepackage{wasysym}
\usepackage{mathtools}
\usepackage{cleveref}
\usepackage{stmaryrd} 
\usepackage{enumerate}
\usepackage{footmisc}

\newcommand{\ineg}{{\smallsmile}}
\newcommand{\uneg}{{\smallfrown}}
\newcommand{\narb}{{-}}
\newcommand{\wnarb}{\scalebox{1}{$\mathrlap{\Circle}\textcolor{black}{\narb}$}}

\newcommand{\wsmile}{\scalebox{1}{$\mathrlap{\Circle}\textcolor{black}{\smallsmile}$}}
\newcommand{\wfrown}{\scalebox{1}{$\mathrlap{\Circle}\textcolor{black}{\smallfrown}$}}
\newcommand{\connec}{\pentagon}
\newcommand{\Circled}[1]{\textcircled{\scriptsize$#1$}}
\newcommand{\ttt}{t}
\newcommand{\fff}{f}
\newcommand{\TT}{\mathbf{T}}
\newcommand{\FF}{\mathbf{F}}
\newcommand{\XX}{\mathbf{X}}
\newcommand{\val}[4]{\ensuremath{{#1}^{#2}_{#3}(#4)}}
\newcommand{\tru}[2]{\val{\TT}{}{#1}{#2}}
\newcommand{\fal}[2]{\val{\FF}{}{#1}{#2}}
\newcommand{\truQ}[2]{\val{\TT}{Q}{#1}{#2}}
\newcommand{\falQ}[2]{\val{\FF}{Q}{#1}{#2}}
\newcommand{\xval}[2]{\val{\XX}{}{#1}{#2}}
\newcommand{\xvalQ}[2]{\val{\XX}{Q}{#1}{#2}}

\newcommand{\setP}{{\cal P}}

\newcommand{\setS}{{\cal L}}
\newcommand{\M}{{\cal M}}

\newcommand{\D}{{\bf D}}
\newcommand{\B}{{\bf B}}
\newcommand{\K}{{\bf K}}
\newcommand{\T}{{\bf T}}
\newcommand{\FUNC}{{\bf Fun}}
\newcommand{\Q}{{\cal Q}}
\newcommand{\F}{{\cal F}}
\newcommand{\E}{{\cal E}}

\newcommand{\wfr}{\prec}
\newcommand{\wfreq}{\preceq}

\DeclareSymbolFont{symbolsC}{U}{txsyc}{m}{n}
\DeclareMathSymbol{\strictif}{\mathrel}{symbolsC}{74}
\newcommand{\suq}{\subseteq}
\newcommand{\set}[1]{\left\{#1\right\}}
\newcommand{\tup}[1]{\left<#1\right>}

\newcommand{\st}{{\ |\ }}   

\newcommand{\Ra}{\Rightarrow}

\newcommand{\rs}{{\,/\,}} 
\newcommand{\GKF}{{\rm {KF}}}
\newcommand{\mostbasic}{{\rm {P}}}
\newcommand{\GPK}{{\rm {PK}}}
\newcommand{\G}{{\rm {G}}}

\newcommand{\GPKB}{{\rm {PKB}}}
\newcommand{\GPKDB}{{\rm {PKDB}}}
\newcommand{\GPKD}{{\rm {PKD}}}
\newcommand{\GPKT}{{\rm {PKT}}}
\newcommand{\GPKF}{{\rm {PKF}}}
\newcommand{\GPKFour}{{\rm {PK4}}}
\newcommand{\GPKFourD}{{\rm {PKD4}}}
\newcommand{\KFour}{{\bf 4}}
\newcommand{\g}{\Gamma}
\renewcommand{\d}{\Delta}
\newcommand{\w}{\wedge}
\newcommand{\ssrul}[2]{\begin{array}{c}#1\\ \hline \ST #2\end{array}}
\newcommand\ST{\rule[-0.5em]{0pt}{1.5em}}

\usepackage{color}

\newcommand{\dom}[1]{{{\textsl dom}}(#1)}
\newcommand{\Lano}[2]{{#1}{{}\!\setminus\!{#2}}}
\newcommand*{\ddfrac}[2]{\genfrac{}{}{0pt}{0}{#1}{#2}} 

%% file: sec-goals.tex
\paragraph{Denying instead of affirming}
Many well-known subclassical logics ---in\-clu\-ding intuitionistic logic and several many-valued logics--- share the con\-junc\-tion-disjunction fragment of classical logic, but disagree about the exact notion of opposition and the specific logical features to be embodied in 
negation.
In contrast, modal logics are often thought of as superclassical, and are obtained by the addition of 
`positive modalities'~$\Box$ and~$\lozenge$.
For various well-known cases, such modalities fail to have a finite-valued characterization.
Notwithstanding, each $m$-ary connective~$\connec$ of a modal logic is typically \textit{congruential} (with respect to the underlying consequence relation~$\vdash$), in treating equivalent formulas as synonymous:
if $\alpha_i\vdash\beta_i$ and $\beta_i\vdash\alpha_i$, for every $1\leq i\leq m$, then $\connec(\alpha_1,\ldots,\alpha_m)\vdash\connec(\beta_1,\ldots,\beta_m)$.
To logical systems containing only such sort of connectives one might associate semantics in terms of neighborhood frames (see ch.5 of~\cite{Wojcicki88}), and the same applies if one uses 1-ary `negative modalities' instead, as in~\cite{ripley:PhD}.
The family of systems enjoying congruentiality (a.k.a.~`replacement property') goes sometimes under the name of `classical modal logics' (cf.~\cite{Segerberg71}).
As a matter of fact, the family of `normal modal logics' makes its 1-ary positive modalities respect a stronger property: if $\alpha\vdash\beta$ then $\connec(\alpha)\vdash\connec(\beta)$.
Such monotone behavior may be captured by semantics based on Kripke frames, and the same applies to the antitone behavior that characterize negative modalities, namely: if $\alpha\vdash\beta$ then $\connec(\beta)\vdash\connec(\alpha)$.

\paragraph{Some of our ancestors}
In~\cite{dun:zho:neggag:05} an investigation of negative modalities is accomplished on top of the ${\land}{\lor}{\top}{\bot}$-
fragment of classical logic, and the same base language had already been considered in~\cite{res:combpossneg:SL97} for the combination of positive and negative modalities.  Typically, in studies of positive and negative modalities the so-called compatibility (bi-relational) frames are used, and certain appropriate conditions upon the commutativity of diagrams involving their two relations are imposed, having as effect the heredity of truth (i.e., its persistence towards the future) with respect to one of the mentioned relations (assumed to be a partial order).
There are a number of studies (e.g.~\cite{vaka:cons89:full,Dos:NMOiIL:1984}) in which the above mentioned languages for dealing with negative modalities are upgraded in order to count on an (intuitionistic or classical) implication, and sometimes also its dual, co-implication (cf.~\cite{rau:BH:DM1980}).  If one may count on classical implication, however, it suffices to add to it the modal paraconsistent negation given by `unnecessity' (cf.~\cite{jmarcos:neNMLiP}), and all other connectives of normal modal logics turn out to be definable already from such an impoverished basis (indeed, where~$\ineg$ is a primitive symbol for unnecessity and~$\supset$ represents classical implication, we have that ${\sim}\alpha:=\alpha\supset\ineg(\alpha\supset\alpha)$ behaves as the classical negation of~$\alpha$, and $\Box\alpha:={\sim}\ineg\alpha$ behaves as the usual positive modality box).
In the particular case of $S5$, an even simpler definition of necessity is within reach (cf.~\cite{bez:JAL:05}), namely $\Box\alpha:=\ineg\ineg\alpha$.

\paragraph{Paraconsistency and paracompleteness}
Our basic intuition about the relation between a paracomplete (a.k.a.~`in\-tui\-tion\-istic-like') negation and a paraconsistent negation is that the former would be expected to be more demanding than the latter, while classical negation had better sit between the two (whenever it also turns out to be expressible).
It takes indeed more effort to assert a negated statement on constructive grounds, while such statements are more readily asserted should some contradictions be allowed to subsist; in other words, one could say that negations in a paracomplete logic come at a greater cost than classical negations, while paraconsistent logics indulge on negations in which classical logic would show greater restraint. 
The presence of a classical negation, however, often makes it too easy to forget that there are two distinct kinds of deviations equally worth studying, concerning non-classical negation, for one of these deviations may then be recovered in the standard way as the dual of the other. 
However, duality does not presuppose definability: in the case of the basic language for positive modal logic (cf.~\cite{dun:PML:1995}), it is well-known that no classical negation is definable, and that the positive modalities are not interdefinable (cf.~\cite{cel:jan:PML:99}).

\paragraph{A richer language in which to study negative modalities}
Assume that~$\neg$ is a 1-ary symbol for negation.  According to the classical `consistency assumption', [CA], there is no state of affairs~$v$ and no formula~$\varphi$ such that [Inc] both~$\varphi$ and~$\neg\varphi$ are satisfied in~$v$.  The dual `determinedness assumption', [DA], has it that there is no state of affairs~$v$ and no formula~$\varphi$ such that [Und] both~$\varphi$ and~$\neg\varphi$ are left unsatisfied in~$v$.  Whenever a logic contains a negation that behaves non-classically, at least one of the above mentioned assumptions is bound to fail.
The so-called `Logics of Formal Inconsistency' (\textbf{LFI}s) provide tools for recovering [CA], by offering in their language a 1-ary connective~$\mathsf{C}$ such that $\mathsf{C}\varphi$ is left unsatisfied in~$v$ whenever [Inc] happens to be the case.  Dually, the `Logics of Formal Undeterminedness' (\textbf{LFU}s) offer a 1-ary connective~$\mathsf{D}$ such that $\mathsf{D}\varphi$ is satisfied in~$v$ whenever [Und] happens to be the case.
The \textbf{LFI}s and \textbf{LFU}s investigated in the present paper are in fact a little bit stronger than that: in them, $\mathsf{C}\varphi$ is left unsatisfied \textit{iff} [Inc] is the case, and $\mathsf{D}\varphi$ is satisfied \textit{iff} [Und] is the case.
The `adjustment' (a.k.a.\ `restoration') connectives $\mathsf{C}$ and~$\mathsf{D}$ are meant to exclude the scenarios in which negation deviates from [CA] and [DA], and to allow for a given reasoner, if that be the case, to recover an intended classical behavior from within a non-classical environment.

\paragraph{On the availability of classical negation}
In order to get a better grasp of the duality between paraconsistent and paracomplete modal negations (namely, unnecessity vs.\ impossibility), we here purposefully make an effort to prevent the underlying language from being sufficiently expressive so as to allow for the definition of  a classical negation (or a classical implication) --- whenever such goal lies within reach.  
As we will see, in fact, all normal modal logics in the basic language of negative modalities 
that happen to fail the consistency and the determinedness assumptions also 
fail to be expressive enough so as to allow for a classical negation to be defined, but in a few cases the definability is within reach if adjustment connectives are employed. 
It should be noted, though, that as a byproduct of the presence of the above mentioned adjustment connectives truth will no longer be hereditary in our Kripke models, that is, it will not in general be preserved for all compound formulas towards the future --- this stands in stark contrast with what happens with models of compatibility frames.

\paragraph{What is to follow}
In this paper, first and foremost we will concentrate on the logic $PK$, determined by the class of all Kripke frames, which has been introduced and received a presentation as a sequent system in~\cite{dod:mar:ENTCS2013}, against a theoretical background of definitions situated at the level of abstract consequence relations.  We show here that this logic can be reintroduced in terms of a so-called `basic sequent system', which allows one to take advantage of general techniques developed in~\cite{lah:avr:Unified2013}, including a method for proving soundness and completeness with respect to the Kripke semantics of $PK$ (given in~\cite{dod:mar:ENTCS2013}), as well as a uniform recipe for semantic proofs of cut-admissibility or analyticity.  After that we apply a similar strategy to the study of several extensions of $PK$ that happen to validate principles distinctive of classical negation, and next we also investigate the explicit definability of classical negation within our logics.
Before devoting ourselves to that task, though, the next section shall adopt a semantical perspective to explain the circumstances in which our study is developed. %

%% file: sec-intro.tex
We briefly recall the now standard components of a Kripke semantics.
A~\textsl{frame} is a structure consisting of a nonempty set~$W$ (of `worlds') and a binary (`accessibility') relation~$R$ on~$W$.
A(n ordinary) \textsl{model} ${\cal M}=\langle {\cal F},V\rangle$ is based on a frame ${\cal F}=\langle W, R\rangle$ and on a \textsl{valuation} $V:W\times \setS \to \{\fff,\ttt\}$ that assigns truth-values
to worlds $w\in W$ and sentences~$\varphi$ of a propositional language~$\setS$ generated over a denumerable set of propositional variables~$\setP$.
The valuations must satisfy certain conditions that are induced by the fixed interpretation of the connectives of the given language.
When $V(w,\varphi)=\ttt$ we say that~$V$ \textsl{satisfies~$\varphi$ at~$w$}, and denote this by ${\cal M}, w\Vdash \varphi$; otherwise 
we write ${\cal M}, w\not\Vdash \varphi$ and say that~$V$ \textsl{leaves~$\varphi$ unsatisfied at~$w$}.
The connectives from the positive fragment of classical logic receive their standard boolean interpretations locally, world-wise, by recursively setting:

\smallskip

\noindent
\begin{tabular}{l l c l}
{[S$\top$]} & ${\cal M}, w\Vdash \top$\\
{[S$\land$]} & ${\cal M}, w\Vdash \varphi\land\psi$ & iff &  ${\cal M}, w\Vdash \varphi$ and ${\cal M}, w\Vdash\psi$\\
{[S$\lor$]} & ${\cal M}, w\not\Vdash \varphi\lor\psi$ & iff &  ${\cal M}, w\not\Vdash \varphi$ and ${\cal M}, w\not\Vdash\psi$\\
\end{tabular}
\smallskip

\noindent
Given formulas $\Gamma\cup\Delta$ of~$\setS$, and given a class of frames~${\cal E}$, we say that \textsl{$\Gamma$ entails $\Delta$ in~${\cal E}$}, and denote this by $\Gamma\models_{\cal E}\Delta$, if for each model~${\cal M}$ based on a frame ${\cal F}\in{\cal E}$ and each world~$w$ of ${\cal M}$ we have either ${\cal M},w\not\Vdash\gamma$ for some $\gamma\in\Gamma$ or ${\cal M},w\Vdash\delta$ for some $\delta\in\Delta$.
The assertion $\Gamma\models_{\cal E}\Delta$ will be called a \textsl{consecution}; when such assertion happens to be true we may also say about each frame in~$\cal E$ that it \textsl{validates} the given consecution. 
In the next section we will extend the notion of entailment so as to cover sequents instead of formulas.
The subscript~${\cal E}$ shall be omitted in what follows whenever there is no risk of ambiguity.

In the following subsections we extend the above language with connectives whose modal interpretations will be useful for the investigation of negations in a non-classical congruential setting.

\subsection{Adding negations}
\label{addingNEG}

Our first extension of the above language proceeds by the addition of a 1-ary connective~$\ineg$, to be interpreted  non-locally (that is, its satisfaction depends on accessible worlds) as follows:
\smallskip

\noindent
\begin{tabular}{l l c l}
{[S$\ineg$]} & ${\cal M}, w\Vdash \ineg\varphi$ & iff & ${\cal M}, v\not\Vdash \varphi$ for \underline{some} $v\in W$ such that $wRv$\\
\end{tabular}
\smallskip

\noindent
Accordingly, a formula $\ineg\varphi$ is said to be satisfied at a given world of a model precisely when the formula $\varphi$ fails to be satisfied at some world accessible from this given world.
In the following paragraph we will show that~$\ineg$ respects some minimal conditions to deserve being called a `negation', namely, we will demonstrate its ability to invert truth-values assigned to certain formulas (at certain worlds).

Let~$\#$ represent an arbitrary 1-ary connective, and let $\#^j$ abbreviate a $j$-long sequence of~$\#$'s.  The least we will demand from~$\#$ to call it a \textsl{negation} is that, for every $p\in\setP$ and every $k\in\mathbb{N}$:
\smallskip

\noindent
\begin{tabular}{l l l l}
{$\llbracket$\textit{falsificatio}$\rrbracket$} &  $\#^k p\not\models\#^{k+1} p$ \hspace{10mm} &
{$\llbracket$\textit{verificatio}$\rrbracket$} &  $\#^{k+1} p\not\models\#^k p$\\
\end{tabular}
\smallskip

\noindent
To witness {$\llbracket$\textit{falsificatio}$\rrbracket$}, some sentence~$\varphi$ is to be satisfied while the sentence $\#\varphi$ is not simultaneously satisfied; for {$\llbracket$\textit{verificatio}$\rrbracket$} some sentence~$\varphi$ is to be left unsatisfied while at the same time $\#\varphi$ is satisfied.
To check that the connective~$\ineg$ fulfills such requisites, it suffices for instance to build a frame in which $W=\{w_n:n\in\mathbb{N}\}$ and $wRv$ iff $v=w^{+\!+}$ (by which we mean that $v$ is the successor of~$w$), and consider a valuation~$V$ such that $V(w_n,p)=\ttt$ iff $n$ is odd.

It is very easy to see that our connective~$\ineg$
satisfies \textsl{global contraposition}
in the sense that
$\alpha\models\beta\mbox{ implies }\ineg\beta\models\ineg\alpha$.
Indeed, assume $\alpha\models\beta$ and suppose that  ${\cal M}, w\Vdash\ineg\beta$ for some world~$w$ of an arbitrary model~${\cal M}$.  Then, [S$\ineg$] informs us that there must be some world~$v$ in~${\cal M}$ such that $wRv$ and ${\cal M}, v\not\Vdash \beta$.  By the definition of entailment, we conclude thus from the initial assumption that ${\cal M}, v\not\Vdash \alpha$.  Using again [S$\ineg$] it follows that ${\cal M}, w\Vdash\ineg\alpha$.  As a byproduct of this, if one defines an equivalence relation~$\equiv$ on~$\setS$ by setting $\alpha\equiv\beta$ whenever both $\alpha\models\beta$ and $\beta\models\alpha$, then an easy structural induction on~$\setS$ establishes that~$\equiv$ is not only compatible with~$\ineg$ but also with the other connectives that are used in constructing the algebra of formulas. This means that~$\equiv$ constitutes a congruence relation on~$\setS$.

It is straightforward to check that any 
1-ary connective~$\#$ satisfying global contraposition is such that, given $p,q\in\setP$:
\smallskip

\noindent
\begin{tabular}{l l l l}
(DM1.1\#) &  $\#(p\lor q)\models\# p\land \# q$ \hspace{5mm} &
(DM2.1\#) &  $\# p\lor \# q\models\# (p\land q)$\\
\end{tabular}
\smallskip

\noindent
If~$\#$ also respects the following consecutions, then it is said to be a \textsl{full type 
diamond-minus connective}:
\smallskip

\noindent
\begin{tabular}{l l l l}
(DM2.2\#) &  $\#(p\land q)\models\# p\lor \# q$ \hspace{5mm} & %
(DT\#) &  $\# \top\models p$\\
\end{tabular}
\smallskip

\noindent
Note that our negation~$\ineg$ is a full type 
diamond-minus connective.
To check that $\ineg$ satisfies (DM2.2\#), indeed, suppose that ${\cal M}, w\Vdash\ineg(p\land q)$ for some arbitrary world~$w$ of an arbitrary model~${\cal M}$.  By [S$\ineg$] we know that there is some world~$v$ such that $wRv$ and ${\cal M}, v\not\Vdash p\land q$.  It follows by [S$\land$] that ${\cal M}, v\not\Vdash p$ or  ${\cal M}, v\not\Vdash q$.  Using [S$\ineg$] again we conclude that ${\cal M}, w\Vdash\ineg p$ or ${\cal M}, w\Vdash\ineg q$ and [S$\lor$] gives us ${\cal M}, w\Vdash\ineg p\lor \ineg q$.  In addition, to check that~$\ineg$ satisfies (DT\#) one may invoke [S$\ineg$] and [S$\top$].  Note that satisfying (DT\#) means that the nullary connective~$\bot$ taken as an abbreviation of $\# \top$ is interpretable by setting, for every world~$w$ of every model~${\cal M}$:
\smallskip

\noindent
\begin{tabular}{l l c l}
{[S$\bot$]} & ${\cal M}, w\not\Vdash \bot$\\
\end{tabular}
\smallskip

Given a negation~$\#$, we call the logic containing it \textsl{$\#$-paraconsistent} if the following consecution fails, for $p,q\in\setP$:
\smallskip

\noindent
\begin{tabular}{l l}
{$\llbracket$\#-explosion$\rrbracket$} &  $p,\#p\models q$\\
\end{tabular}
\smallskip

\noindent
This means that there must be valuations that satisfy both some sentence~$\varphi$ and the sentence $\#\varphi$ while not satisfying every other sentence.  It is worth noticing that $\llbracket\ineg$-explosion$\rrbracket$ holds good in frames containing exclusively worlds that are accessible to themselves, and to themselves only (such worlds will be called `narcissistic') and worlds that do not access any other world (such worlds will be called `dead ends'): in the former case, it is impossible to simultaneously satisfy both~$\varphi$ and $\ineg\varphi$; in the latter case, the sentence $\ineg\varphi$ is never satisfied.  Note moreover that in the class of all narcissistic frames (those containing only narcissistic worlds) the connective~$\ineg$ happens to behave like \textsl{classical negation}, i.e., it behaves like the symbol~$\sim$ in the following semantic clause:
\smallskip

\noindent
\begin{tabular}{l l c l}
{[S$\sim$]} &  ${\cal M},w\Vdash {\sim} \varphi$ & iff & ${\cal M},w\not\Vdash \varphi$ \\
\end{tabular}
\smallskip

\noindent
In contrast, in the class of all frames whose worlds are all dead ends the connective~$\ineg$ does not respect [\textit{verificatio}], and cannot be said thus to constitute a negation.

We now make a further extension of the above language by adding a 1-ary connective $\uneg$, non-locally interpreted as follows:
\smallskip

\noindent
\begin{tabular}{l l c l}
{[S$\uneg$]} & ${\cal M}, w\Vdash \uneg\varphi$ & iff & ${\cal M}, v\not\Vdash \varphi$ for \underline{\smash{every}} $v\in W$ such that $wRv$\\
\end{tabular}
\smallskip

\noindent
It is not difficult to check that again we have a connective that qualifies as a negation, and satisfies global contraposition.
\noindent
To reinforce the meta-theo\-re\-tical duality between the latter negation and the negation introduced above through [S$\ineg$], we will henceforth refer to the previous interpretation clause in the following equivalent form:
\smallskip

\noindent
\begin{tabular}{l l c l}
{[S$\uneg$]} & ${\cal M}, w\not\Vdash \uneg\varphi$ & iff & ${\cal M}, v\Vdash \varphi$ for \underline{some} $v\in W$ such that $wRv$\\
\end{tabular}
\smallskip

\noindent
A \textsl{full type 
box-minus connective} is a 1-ary connective~$\#$ that respects:
\smallskip

\noindent
\begin{tabular}{l l l l}
(DM1.2\#) &  $\# p\land \# q\models\#(p\lor q)$ \hspace{5mm} & %
(DF\#) &  $p\models \# \bot$\\
\end{tabular}
\smallskip

\noindent
One may easily check that~$\uneg$ is indeed a full type 
box-minus connective.

Given a negation~$\#$, we call the logic contaning it \textsl{$\#$-paracomplete} if it fails the following consecution, for $p,q\in\setP$:
\smallskip

\noindent
\begin{tabular}{l l}
{$\llbracket$\#-implosion$\rrbracket$} &  $q\models \# p, p$\\
\end{tabular}
\smallskip

\noindent
Such failure will clearly be the case for~$\#=\uneg$ as soon as we entertain frames that contain worlds that are neither dead ends nor narcissistic.  Otherwise, we see that $\uneg$ will behave either like classical negation (if all worlds are narcissistic) or like~$\top$ (if all worlds are dead ends).

In the following sections, unless noted otherwise, we will no longer consider classes of frames containing only frames with worlds that are either dead ends or narcissistic --- accordingly, we will only consider entailment relations that are $\ineg$-paracon\-sist\-ent and $\uneg$-paracomplete, for the negative modalities $\ineg$ (assumed to be full-type diamond-minus) and~$\uneg$ (assumed to be full-type box-minus).

\subsection{Recovering negation-consistency and negation-determinedness}
\label{recovering}

In what follows we will call a model \textsl{dadaistic} when it contains some world in which all formulas are satisfied, and call it \textsl{nihilistic} if it leaves all formulas unsatisfied at some world.  It is straightforward to see that the language based on ${\land}{\lor}\top\ineg$, with the above interpretations, admits dadaistic models, while the language based on ${\land}{\lor}\bot\uneg$ admits nihilistic models.

Recall that a $\#$-paraconsistent logic allows for valuations that satisfy certain formulas~$\varphi$ and~$\# \varphi$ while leaving some other formula~$\psi$ unsatisfied (at some fixed world).  There might be reasons for disallowing this phenomenon to occur with an arbitrary~$\varphi$, or for restricting to certain formulas~$\psi$ but not others.  A particularly useful way of keeping a finer control over which `inconsistencies' of the form~$\varphi$ and~$\# \varphi$ are to be acceptable within non-dadaistic models is to mark down the formula thereby involved so as to recover a `gentle' version of $\llbracket$\#-explosion$\rrbracket$.  
Concretely, for us here, a 1-ary `adjustment connective'~$\Circled{\#}$ that strongly internalizes the meta-theoretic consistency assumption at the object language level will be such that:
\smallskip

\noindent
\begin{tabular}{l l c l}
{[SC$\#$]} & ${\cal M}, w\Vdash \Circled{\#}\varphi$ & iff & ${\cal M}, w\not\Vdash \varphi$ or ${\cal M}, w\not\Vdash \#\varphi$\\
\end{tabular}
\smallskip

\noindent
It is easy to check that any connective~$\Circled{\#}$ respecting [SC$\#$] is such that:
\smallskip

\noindent
\begin{tabular}{l l l l l l}
(C1\#) &  $\Circled{\#} p, p, \#  p\models$ \hspace{5mm} &
(C2\#) &  $\models p,\Circled{\#} p$ \hspace{5mm} &
(C3\#) &  $\models \# p,\Circled{\#} p$ \\
\end{tabular}
\smallskip

\noindent
Note in particular that (C1\#) guarantees that there are no valuations that satisfy (at a fixed world) both $p$ and $\# p$ if these are put in the presence of $\Circled{\#} p$.  Thus, in case $\#$ fails $\llbracket$\#-explosion$\rrbracket$ we may look at the latter formula involving~$\Circled{\#}$ as guaranteeing that a weaker form of explosion is available.  On these grounds we shall call the connective~$\Circled{\#}$ an \textsl{adjustment companion} to~$\#$: it allows one to recover explosion from within a non-\#-explosive (i.e., paraconsistent) logical context, and adjust the consecutions of the underlying logic so as to allow for the simulation of the consecutions that would otherwise be justified by reference to $\llbracket$\#-explosion$\rrbracket$.  Semantically, the presence of such connective also guarantees that dadaistic models are not admissible over the language based on ${\land}{\lor}\top\ineg\wsmile$, with the above interpretations.  This is because a formula of the form $\wsmile\varphi\land(\varphi\land\ineg\varphi)$ is equivalent to a formula~$\bot$ respecting [S$\bot$].

Dually, a \#-paracomplete logic allows for valuations that leave the formulas~$\varphi$ and~$\#\varphi$ both unsatisfied (at some fixed world), while satisfying some other formula~$\psi$. A particular way of keeping a finer control over which `indeterminacies' of the form $\varphi$ and~$\#\varphi$ are to be acceptable within non-nihilistic models is to allow for a `gentle' version of $\llbracket${\#-implosion}$\rrbracket$, where a 1-ary connective~$\Circled{\#}$ internalizes the meta-theoretic determinedness assumption at the object language level, in such a way that:
\smallskip

\noindent
\begin{tabular}{l l c l}
{[SD$\#$]} & ${\cal M}, w\not\Vdash \Circled{\#}\varphi$ & iff & ${\cal M}, w\Vdash \varphi$ or ${\cal M}, w\Vdash \#\varphi$\\
\end{tabular}
\smallskip

\noindent
Clearly, any connective~$\Circled{\#}$ respecting [SD$\#$] is such that:
\smallskip

\noindent
\begin{tabular}{l l l l l l}
(D1\#) &  $\models\#  p, p, \Circled{\#} p$ \hspace{5mm} &
(D2\#) &  $\Circled{\#} p,p\models $ \hspace{5mm} &
(D3\#) &  $\Circled{\#} p,\# p\models$ \\
\end{tabular}
\smallskip

Note that a formula of the form $(\varphi\lor\uneg\varphi)\lor\wsmile\varphi$ is equivalent to a formula~$\top$ respecting [S$\top$].
Note, moreover, that whenever it turns out that a connective~$\#$ respects $\llbracket${\#-explosion}$\rrbracket$ and at the same time its adjustment companion~$\Circled{\#}$ re\-spects [SC$\#$], then the formula $\Circled{\#}\varphi$ is equivalent to~$\top$.  In an analogous way, whenever a connective~$\#$ respects $\llbracket$\#-implosion$\rrbracket$ and at the same time its adjustment companion~$\Circled{\#}$ respects [SD$\#$], the formula $\Circled{\#}\varphi$ is equivalent to~$\bot$.  This stresses the fact that the adjustment connectives with which we deal in this subsection are of more interest when they accompany the respective non-classical negations to whose meaning they contribute.
\smallskip

At this point we have finally finished constructing the richest language that will be used throughout the rest of the paper: It will contain the connectives ${\land}{\lor}\top\bot\ineg\wsmile\uneg\wfrown$, disciplined by the corresponding [S\#] conditions above.  In the following subsection we will explain precisely when a classical negation, that is a 1-ary connective~$\sim$ subject to condition [S${\sim}$], is \textit{definable} with the use of our language.
Fixed such language, the logic characterized over it~by the class~$\E$ of all frames will be called $PK$; the logic characterized by~the class~$\E_{\D}$ of all serial frames (the frames with serial accessibility relations) will be called $PKD$; the logic characterized by the class~$\E_{\T}$ of all reflexive frames will be called $PKT$; the logic characterized by the class~$\E_{\FUNC}$ of all frames whose accessibility relations are total functions will be called $PKF$; the logic characterized by the class~$\E_{\B}$ of all symmetric frames will be called $PKB$; the logic characterized by the class~$\E_{\KFour}$ of all transitive frames will be called $PK4$; and by $PKD4$ and $PKDB$ we will refer to the logics characterized by the classes of all frames whose accessibility relations enjoy a combination of the two obvious associated properties, in each case.

\subsection{Around classical negation}
\label{aroundCN}

According to the intuitions laid down at \Cref{sec:goals}, one could expect that in general (a) $\uneg\alpha\vdash{\sim}\alpha$ and (b) ${\sim}\alpha\vdash\ineg\alpha$.  It is easy to see, for instance, that these consecutions are indeed sanctioned by $PKT$, for the classical negation~$\sim$ that may be defined by setting ${\sim}\varphi:=\ineg\varphi\land\wsmile\varphi$ (alternatively, one may set ${\sim}\varphi:=\uneg\varphi\lor\wfrown\varphi$).
Indeed, suppose on the one hand that ${\cal M}, w\Vdash\uneg\alpha$ for a world~$w$ of a model~${\cal M}$ of a reflexive frame.  By reflexivity and [S$\uneg$] we must have ${\cal M}, w\not\Vdash\alpha$, and this is equivalent to ${\cal M}, w\Vdash{\sim}\alpha$ by [S$\sim$].  On the other hand, given ${\cal M}, w\not\Vdash\alpha$, by reflexivity and [S$\ineg$] we immediately conclude that ${\cal M}, w\Vdash\ineg\alpha$.

Meanwhile, in the deductively weaker logic $PKD$ one cannot in general prove neither (a) nor (b), even though a classical negation may be defined in this logic by setting ${\sim}\varphi:=(\uneg\varphi\land\wsmile\varphi)\lor\wfrown\varphi$.  However, one can still easily prove in $PKD$ that (c)~$\uneg\alpha\vdash\ineg\alpha$.
Indeed, suppose that ${\cal M}, w\Vdash\uneg\alpha$.  By seriality and [S$\uneg$] we conclude that ${\cal M}, v\not\Vdash\alpha$ for some world~$v$ such that $wRv$.  But then, by [S$\ineg$] it follows that ${\cal M}, w\Vdash\ineg\alpha$.
In the logic $PKF$, deductively stronger than $PKD$ (but neither stronger nor weaker than $PKT$) one may also prove the converse consecution, (d) $\ineg\alpha\vdash\uneg\alpha$.
Indeed, suppose that ${\cal M}, w\Vdash\ineg\alpha$.  There is, by the fact that the accessibility relation is a total function, a single world~$v$ such that $wRv$.  Then ${\cal M}, v\not\Vdash\alpha$, by [S$\ineg$]. For a similar reason, invoking now [S$\uneg$] we conclude that ${\cal M}, w\Vdash\uneg\alpha$.
Note that (c) and (d) together make our two modal non-classical negations indistinguishable from the viewpoint of $PKF$, yet one should not for this reason imagine, as we will see, that they collapse into classical negation.

The situation concerning classical negation and its relation to its non-clas\-sical neighbours gets even more interesting after one acknowledges that \textit{no}~classical negation is definable in $PK$, the weakest of our logics, and also that \textit{no} classical ne\-gation is definable in the fragment of $PKT$ without neither of the adjustment connectives, nor in the fragment of $PKF$ (or $PKD$) without either one of the adjustment connectives, nor in $PKB$ or in $PK4$.  Detailed proofs concerning the mentioned results about (non)definability of classical negation in the modal logics that constitute our present object of study may be found in \Cref{sec:definability}. 

Notice that in $PKD$ and its extensions 
there are no negated formulas that happen to be true or false at a given world just because there are no worlds accessible from it.
Note also that the logic $PKT$ is: paraconsistent but not paracomplete with respect to the connective~$\ineg$; paracomplete but not paraconsistent with respect to~$\uneg$
(adapting the result in \cite{avr:zam:AiML16}, it may be shown that this logic is indeed the least extension of the positive implicationless fragment of classical logic with the latter mentioned properties).
In all the other logics mentioned above, in contrast, both non-classical negations behave at once as paracomplete and paraconsistent negations (recall, though, that each negation is associated to a different adjustment connective).
We take the cases among these in which no classical negation is available to be particularly attractive for the task of revealing the `uncontamined' nature of non-classical negation.  Establishing well-behaved proof theoretical counterparts for such logics, as we shall do in what follows, is meant to allow for them to be even better understood and dealt with.


%% file: sec-proofsystem.tex

A sequent calculus for $PK$, that we denote by $\GPK$, was introduced in~\cite{dod:mar:ENTCS2013}, and consists of the following rules:
%
{\small
\[\begin{array}{ll@{\hspace{2em}}ll}
{[id]} & \ssrul{}{p\Ra p} &
 {[cut]} & \ssrul{\g,\varphi\Ra \d \ \ \ \ \g\Ra \varphi,\d}{\g\Ra \d}
\\[2mm]
{[W{\Ra}]} & \ssrul{\g\Ra \d}{\g,\varphi\Ra \d}
\qquad & {[{\Ra}W]} & \ssrul{\g\Ra\d }{\g\Ra \varphi,\d} \\
\end{array}\]}
{\small
\[\begin{array}{ll@{\hspace{2em}}ll}
{[{\bot}{\Ra}]} & \ssrul{}{\g,\bot\Ra \d} &
{[{\Ra}{\top}]} & \ssrul{}
{\g\Ra\top,\d}
\\
{[{\w}{\Ra}]} & \ssrul{\g,\varphi,\psi\Ra \d}{\g,\varphi\w \psi\Ra \d} &
{[{\Ra}{\w}]} & \ssrul{\g\Ra \varphi,\d\ \ \ \g\Ra \psi,\d}
{\g\Ra \varphi\w \psi,\d}
\\
{[{\vee}{\Ra}]} & \ssrul{\g,\varphi\Ra \d\ \ \ \g,\psi\Ra \d}
{\g,\varphi\vee \psi\Ra \d} &
{[{\Ra}{\vee}]} & \ssrul{\g\Ra \varphi, \psi,\d}{\g\Ra \varphi\vee \psi,\d}
\\[3mm]
{[{\ineg}{\Ra}]} & \ssrul{\g\Ra\varphi, \d}
{\uneg\d,\ineg\varphi\Ra \ineg\g} &
{[{\Ra}{\uneg}]} &
\ssrul{\g,\varphi\Ra \d}{\uneg\d\Ra \uneg\varphi,\ineg\g}
\end{array}\]
\[\begin{array}{ll@{\hspace{2em}}ll}
{[{\wsmile}{\Ra}]} & \ssrul{\g\Ra\varphi, \d\ \ \ \g\Ra\ineg\varphi,\d}
{\g,\wsmile\varphi\Ra\d} &
{[{\Ra}{\wsmile}]} &
\ssrul{\g,\varphi,\ineg\varphi\Ra \d}{\g\Ra\wsmile\varphi,\d}
\\[1mm]
{[{\wfrown}{\Ra}]} & \ssrul{\g\Ra\varphi,\uneg\varphi,\d}
{\g,\wfrown\varphi\Ra\d} &
{[{\Ra}{\wfrown}]} &
\ssrul{\g,\varphi\Ra \d\ \ \ \g,\uneg\varphi\Ra\d}{\g\Ra\wfrown\varphi,\d}
\end{array}\]
}

\noindent
Above, sequents are taken to have the form $\Sigma\Ra\Pi$ where~$\Sigma$ and~$\Pi$ are finite sets of formulas, and given a unary connective~$\#$ and $\Psi\subseteq\mathcal{L}$, by $\#\Psi$ we denote the set $\set{\#\psi\st\psi\in\Psi}$.
We write $S\vdash_{\GPK}s$ to say that there is a derivation in $\GPK$ of a sequent~$s$ from a set~$S$ of sequents.
That establishes a consequence relation between sequents.
A consequence relation between formulas is defined by setting
$\g\vdash_{\GPK}\varphi$ if $\vdash_{\GPK}\g'\Ra\varphi$ for some finite subset~$\g'$ of~$\g$.
The overloaded notation  $\vdash_{\GPK}$ will always be resolved by the pertinent context.
%
Moreover, using a straightforward induction, it is easy to verify that we have $\varphi\Ra\varphi$ for every formula~$\varphi$, and not just for atomic formulas.
This property, called `axiom-expansion' in \cite{Avron:2009} and `id-inductivity' in \cite{sep-connectives-logic}, is sometimes considered important, when designing sequent systems, for a unique characterization of the connectives by a collection of rules. 
It is worth pointing out, nevertheless, that various sequent systems for paraconsistent logics--- e.g., the systems in \cite{Avron2015219}--- do not enjoy this property.

We utilize in what follows the general mechanisms and techniques applicable to the so-called `basic systems' of~\cite{lah:avr:Unified2013}
in order to prove soundness, completeness and cut-admissibility.
Actually, we note that the aforementioned ``axiom expansion" can be also proven using these techniques, however, a simple induction suffices for the particular systems in the present paper.
In~\cite{lah:avr:Unified2013} one may also find a sufficient condition for a  semantic framework to define the same logic as a given basic system. So, instead of using induction on derivations in $\GPK$ for checking soundness, and constructing a canonical model and a maximal theory for completeness, we just {\em write} $\GPK$ as a basic system, and {\em verify} that our models from Section~\ref{sec:intro} satisfy this condition. 

The viewpoint of basic sequent systems enforces the usual distinction between side formulas and principal formulas in sequents. Thus, each sequent is seen as a combination of a `main sequent' (that includes the principal formula) and a `context sequent' (that includes the side formulas).
For an example, in ${[{\Ra}{\vee}]}$, the main sequent of the premise is $\Ra\varphi,\psi$; the main sequent of the conclusion is $\Ra\varphi\vee\psi$; and the context sequent of both is $\g\Ra\d$. Note that in the rules for $\ineg$ and $\uneg$, the context sequent of the premise is different from the context sequent of the conclusion.
Accordingly, \cite{lah:avr:Unified2013} introduces the notion of a \textsl{basic rule}.  Each premise in a basic rule takes the form $\tup{s;\pi}$, where~$s$ is a sequent that corresponds to the main sequent of the premise, and $\pi$ is a relation between singleton-sequents (that is, sequents of the form $\varphi\Ra$ or $\Ra\varphi$), called a \textsl{context relation}, that determines the behavior of the context sequents.
The sequent calculus $\GPK$ may be naturally regarded as a basic system that employs two context relations, namely:
$\pi_0=\set{\tup{q_1\Ra\;;\;q_1\Ra},\tup{\Ra q_1\;;\;\Ra q_1}}$, and
$\pi_1=\set{ \tup{q_1\Ra\;;\;\Ra\ineg q_1}, \tup{\Ra q_1\;;\;\uneg q_1\Ra}}$.
%
The rules of $\GPK$ may then be presented as particular instances of basic rules.
For example, the following are the basic rules for $\wedge,\ineg,\uneg$ and~$\wsmile$:
\smallskip

\noindent
{\small
\(\begin{array}{ll@{\hspace{1.5em}}ll}
{[{\Ra}{\wedge}]} & \tup{\Ra p_1;\pi_0},\tup{\Ra  p_2;\pi_0}\rs \Ra p_1\wedge p_2 &
{[{\wedge}{\Ra}]} & \tup{p_1, p_2\Ra;\pi_0}\rs p_1\wedge p_{2}\Ra \\
{[{\ineg}{\Ra}]} & \tup{\Ra p_1;\pi_1}\rs\ineg p_1\Ra &
{[{\Ra}{\uneg}]} & \tup{p_1\Ra ;\pi_1}\rs\Ra\uneg p_1 \\
{[{\wsmile}{\Ra}]} & \tup{\Ra p_1;\pi_0},\tup{\Ra \ineg p_1;\pi_0}\rs\wsmile p_1\Ra &
{[{\Ra}{\wsmile}]} & \tup{p_1,\ineg p_1\Ra;\pi_0}\rs\Ra\wsmile p_1
\end{array}\)
}
\smallskip

\noindent
The system $\GPK$ employs only one context relation besides $\pi_{0}$, and this context relation uses only one atomic formula $q_{1}$. Moreover, each rule employs exactly one of the context relations. In fact, all systems that we explore in this paper satisfy this property, and their additional context relation is a superset of~$\pi_{1}$ (the only exception being $\GPKF$, below, in which this is only true after~$\ineg$ and~$\uneg$ are identified).
We therefore actually use only a fraction of the full generality \cite{lah:avr:Unified2013} provides.
We shall now use adapted definitions and results that follow from
the latter paper, with the terminology refurbished for this
special case that we are in.

In general, a single\-ton-sequent $x$ will be said to \textsl{relate to a single\-ton-sequent $y$ with respect to a context relation~$\pi$} if there are single\-ton sequents $x'$ and $y'$ such that~$x$ is obtained from~$x'$ and~$y$ is obtained from~$y'$ by replacing $q_{1}$ with some formula~$\varphi$. 
By extension, a sequent will be said to relate to another sequent with respect to~$\pi$ if both sequents can be written as unions of single\-ton-sequents that appropriately relate to one another with respect to~$\pi$.

In applications of $[{\Ra}{\wsmile}]$, the context sequent is left unchanged, as two sequents relate to each other with respect to~$\pi_{0}$ iff they are the same.
In contrast, applications of $[{\ineg}{\Ra}]$ are based on~$\pi_{1}$. Note that a sequent $\g_{1}\Ra\d_{1}$ relates to a sequent $\g_{2}\Ra\d_{2}$ with respect to~$\pi_{1}$ iff $\g_{2}=\uneg\d_{1}$ and $\d_{2}=\ineg\g_{1}$.

Basic systems are endowed with a Kripke semantics. 
The notion of \textit{satisfaction} from Section~\ref{sec:intro} is extended to 
sequents by setting $\M,w\Vdash \g\Ra\d$ if $\M,w\not\Vdash\gamma$ for some $\gamma\in\g$ or $\M,w\Vdash\delta$ for some $\delta\in\d$; in a similar fashion we may talk now about \textit{valid} sequents.
In the general case, each context relation is associated with a certain accessibility relation in frames. For the case of $\pi_{0}$, we take here the trivial relation that consists solely of loops. Actually, this is done implicitly. We simply define the semantic constraints that are associated with this relation {\em locally} in each world. For $\pi_{1}$, on the other hand, we associate the usual accessibility relation found in the models of our frames. 
We stress that the full power of \cite{lah:avr:Unified2013}, that goes far beyond what we need here, 
would require more than one accessibility relation in a single frame. 
--- to wit, the condition $[RR_{r'}]$ (presented below) would need to be verified for each accessibility relation that is associated with the context relations that it employs; furthermore, the condition $[RC_{\pi}]$ (also presented below) would need to be verified for each context relation that occurs in the rules.

For a given basic system $\G$ that employs only one context relation~$\pi$ in addition to $\pi_{0}$, Definitions 4.5 and 4.12 of \cite{lah:avr:Unified2013} impose what we will call here a `$\G$-legal model'.
A model is said to be \textsl{$\G$-legal} if it 
`respects' the basic rules and context relations that constitute~$\G$, where:  \\
${[RR_{r}]}$ \textsl{respecting a basic rule~$r$ that utilizes only $\pi_{0}$} amounts to guaranteeing that in each world the main sequent of the conclusion of~$r$ is satisfied whenever all main sequents of the premises of~$r$ are satisfied;  \\
${[RR_{r'}]}$ \textsl{respecting a basic rule $r'$ that utilizes only some $\pi\neq\pi_0$} amounts to guaranteeing that in each world the main sequent of the conclusion of~$r'$ is satisfied whenever all main sequents of the premises of~$r'$ are satisfied at all accessible worlds; and \\
${[RC_{\pi}]}$ \textsl{respecting a context relation $\pi$} means that the satisfaction of a singleton sequent~$x$ at a world~$u$ implies the satisfaction of a singleton sequent $y$ at a world $w$ whenever~$u$ is accessible from~$w$ and~$x$ relates to~$y$ with respect to~$\pi$.

For an example,  $[RR_{[{\ineg}{\Ra}]}]$  induces the condition:
``if $\M,v\Vdash\;\Ra\varphi$ for every world $v$ such that $w R v$, then $\M,w\Vdash\ineg\varphi\Ra$'',
and this is equivalent to: ``If ${\cal M},w\Vdash \ineg\varphi$  then  ${\cal M}, v\not\Vdash \varphi$ for some $v\in W$ such that $wRv$''.  This amounts to half of clause [S$\ineg$], from Section~\ref{sec:intro}.
Furthermore, $[RC_{\pi_{1}}]$  induces an additional semantic condition: ``if $w R v$ then $\M,w\Vdash\; \Ra\ineg\varphi$ whenever $\M,v\Vdash\varphi\Ra$''.
This amounts to the other half of clause [S$\ineg$], namely:
``${\cal M},w\Vdash \ineg\varphi$  whenever  ${\cal M}, v\not\Vdash \varphi$ for some $v\in W$ such that $wRv$''.
Systematically applying this semantic reading to all rules and all context relations of $\GPK$,
we obtain
the class of all models $\tup{\F,V}$, where~$\F$ is an arbitrary frame and each valuation $V:W\times\setS\rightarrow\set{\fff,\ttt}$ respects the following conditions, for every $w\in W$ and $\varphi,\psi\in\setS$:

\smallskip
\noindent
\begin{tabular}{rl}
	{[$\TT\top$]} & $\tru{w}{\top}$\\[.5mm]
	{[$\FF\bot$]} & $\fal{w}{\bot}$\\[.5mm]
\end{tabular} \\
\begin{tabular}{rl}
	{[$\TT\land$]} & if $\tru{w}{\varphi}$ and $\tru{w}{\psi}$, then $\tru{w}{\varphi\w\psi}$\\
	{[$\FF\land$]} & if $\fal{w}{\varphi}$ or $\fal{w}{\psi}$, then $\fal{w}{\varphi\w\psi}$
	\\[.5mm]
\end{tabular} \\
\begin{tabular}{rl}
	{[$\TT\lor$]} & if $\tru{w}{\varphi}$ or $\tru{w}{\psi}$, then $\tru{w}{\varphi\vee\psi}$\\
	{[$\FF\lor$]} & if $\fal{w}{\varphi}$ and $\fal{w}{\psi}$, then $\fal{w}{\varphi\vee\psi}$
	\\[.5mm]
\end{tabular} \\
\begin{tabular}{rl}
	{[$\TT\ineg$]} & if $\fal{v}{\varphi}$ for some $v\in W$ such that $wRv$, then $\tru{w}{\ineg\varphi}$\\
	{[$\FF\ineg$]} & if $\tru{v}{\varphi}$ for every $v\in W$ such that $wRv$, then $\fal{w}{\ineg\varphi}$
	\\[.5mm]
\end{tabular} \\
\begin{tabular}{rl}
	{[$\TT\uneg$]} & if $\fal{v}{\varphi}$ for every $v\in W$ such that $wRv$, then $\tru{w}{\uneg\varphi}$\\
	{[$\FF\uneg$]} & if $\tru{v}{\varphi}$ for some $v\in W$ such that $wRv$, then $\fal{w}{\uneg\varphi}$
	\\[.5mm]
\end{tabular} \\
\begin{tabular}{rl}
	{[$\TT\wsmile$]} & if $\fal{w}{\varphi}$ or $\fal{w}{\ineg\varphi}$, then $\tru{w}{\wsmile\varphi}$\\
	{[$\FF\wsmile$]} & if $\tru{w}{\varphi}$ and $\tru{w}{\ineg\varphi}$, then $\fal{w}{\wsmile\varphi}$
	\\[.5mm]
\end{tabular} \\
\begin{tabular}{rl}
	{[$\TT\wfrown$]} & if $\fal{w}{\varphi}$ and $\fal{w}{\uneg\varphi}$, then $\tru{w}{\wfrown\varphi}$\\
	{[$\FF\wfrown$]} & if $\tru{w}{\varphi}$ or $\tru{w}{\uneg\varphi}$, then $\fal{w}{\wfrown\varphi}$\\
\end{tabular}
\medskip

\noindent
where we take `$\tru{u}{\alpha}$' as abbreviating `$V(u,\alpha)=\ttt$', and `$\fal{u}{\alpha}$' as abbreviating `$V(u,\alpha)=\fff$'.  If alternatively one just \textit{rewrites} $V(v,\alpha)=\ttt$ as ${\cal M}, v\Vdash \alpha$ and $V(v,\alpha)=\fff$ as ${\cal M}, v\not\Vdash \alpha$, where ${\cal M}=\tup{\tup{W,R},V}$, what results thereby is a collection of conditions that are essentially identical to the [S\#] clauses introduced in our Section~\ref{sec:intro}.
Thus, every model is $\GPK$-legal.

Fix in what follows a model ${\cal M}=\tup{\tup{W,R},V}$.
We say that $w,v\in W$ \textsl{agree with respect to the formula~$\alpha$, according to~$V$}, if
either (both $\tru{w}{\alpha}$ and $\tru{v}{\alpha}$) or else (both $\fal{w}{\alpha}$ and $\fal{v}{\alpha}$).
We say that ${\cal M}$ is \textsl{differentiated} if we have $w=v$ whenever~$w$ and~$v$ agree with respect to every $\alpha\in{\cal L}$, according to~$V\!$.
Now also fix a basic system $\G$ that employs only $\pi_{0}$ and some other context relation $\pi$.
We say that  a $\G$-legal model ${\cal M}$ is  \textsl{$\G$-strengthened}  if  the converse of ${[RC_{\pi}]}$ holds for~$V\!$.  It is worth stressing that the accessibility relation of a $\G$-strengthened model is uniquely determined by the underlying collection of worlds and valuation of this model.

In the case of $\GPK$, we note that $\GPK$-strengthened means that $w R v$ if and only if: (i) $\fal{v}{\varphi}$ implies $\tru{w}{\ineg\varphi}$ and (ii) $\tru{v}{\varphi}$ implies $\fal{w}{\uneg\varphi}$.

\begin{theorem}[Corollary~4.26 in \cite{lah:avr:Unified2013}]
\label{OrisTheorem}
Every basic system $\G$ that employs the context relation~$\pi_{0}$ and some context relation~$\pi$ such that $\pi_{1}\suq\pi$ is sound and complete with respect to any class of  $\G$-legal models that
contains all  differentiated $\G$-strengthened models.
\end{theorem}

Since the class of $\GPK$-legal models is the same as the class of models, we have in particular that the class of models is a class of $\GPK$-legal models that includes all  differentiated 	$\GPK$-strengthened  models.      
The following result from \cite{dod:mar:ENTCS2013} comes as a byproduct:

\begin{corollary}\label{soundnessAndCompletenessForSmiles}
$\g\models_{\E}\varphi$ iff $\g\vdash_{\GPK}\varphi$ for every $\g\cup\set{\varphi}\suq\setS$, where ${\cal E}$ denotes the class of all frames.
\end{corollary}

\noindent 
Indeed, \Cref{OrisTheorem} provides a mechanism that will be conveniently reutilized in the \Cref{sec:extensions}, when we consider extensions of $\GPK$.  

Two brief comments are in order here.
First, our valuation functions
assign truth-values to \emph{every} formula in every world.  However, as the values of compound formulas are uniquely determined by the values of their subformulas, we could have rested content above with assigning truth-values to propositional variables.
Second, given that for the above valuations $\tru{u}{\alpha}$ is the case iff $\fal{u}{\alpha}$ fails to be the case, the semantic conditions [$\TT\#$] and [$\FF\#$], for each connective~$\#$, are clearly the converse of each other.
In setting the two conditions apart,
we have just given them directionality, pointing from less complex to more complex formulas, and have separated between conditions induced by rules from those induced by context relations.
While neither of these manoeuvres are very useful here, they will allow us to more easily relate, in Section~\ref{sec:analyticity}, valuations to `quasi valuations' that have non-truth-functional semantics.
In what follows we will first show that $\GPK$ is proof-theoretically very well-behaved, and have a look next at the extensions of $\GPK$.

%% file: sec-analyticity.tex

In this section we make further use of the powerful machinery introduced in~\cite{lah:avr:Unified2013} to prove that PK enjoys strong cut-admissibility,
in other words, we show that $S\vdash_{\GPK}s$ implies that there is a derivation in $\GPK$ of the sequent~$s$ from the set of sequents~$S$ such that in every application of the cut rule the cut formula~$\varphi$ appears in~$S$.
In particular, $\vdash_{\GPK}s$ implies that~$s$ is derivable in $\GPK$ without any use of the cut rule.

Cut-admissibility in sequent calculi is traditionally proved `through syntactical means', using induction on the derivation length and the complexity of the cut formulas (see, e.g., \cite{Ge34}).
However, such proofs are usually tedious and error-prone.
One of the major contributions of \cite{lah:avr:Unified2013}, that is used here as a substitute of a `syntactic proof', is a \textit{semantic} criterion for cut-admissibility, that allows for a smoother semantic proof of cut-admissibility. Moreover, the modularity of the semantic approach will make it straightforward to adapt the results of the present section to other variants of $\GPK$, that are investigated in the subsequent section. 

The proof is done in two steps. \textit{First}, we present an adequate
semantics for the cut-free fragment of $\GPK$.
\textit{Second}, we show that a countermodel in this new semantics entails the existence of a countermodel in the form of a Kripke model as defined in the previous section. This, together with Corollary~\ref{soundnessAndCompletenessForSmiles}, entails that $\GPK$ is equivalent to its cut-free fragment.

\subsection*{Step 1. Semantics for cut-free $\GPK$}

Semantics for cut-free basic systems may be obtained
through the use of `quasi valuations'.
Models based on quasi valuations differ from ordinary models in two main aspects:
$(a)$ the underlying interpretation is three-valued;
$(b)$ the underlying interpretation is non-deterministic --- the truth-value of a compound formula in a given world is \textit{not} always uniquely determined by the truth values of its subformulas in the collection of worlds of the underlying frame.

Concretely,
given a basic system $\G$ with a context relation $\pi$, and a frame $\F=\tup{W,R}$, a \textsl{quasi valuation} over $\F$ is a function $QV:W\times\setS\rightarrow \set{\set{\fff},\set{\ttt},\set{\fff,\ttt}}$. We call $\tup{\F,QV}$ a \textsl{quasi model}. 
We say that~$QV$ \textsl{satisfies~$\varphi$ at~$w$}, and denote this by ${\cal M}, w\Vdash \varphi$ if  $\ttt\in QV(w,\varphi)$ (instead of $QV(w,\varphi)=\ttt$, that we would have had in case this was an ordinary model).
$\tup{\F,QV}$ is $\G$-legal if it respects all rules of $\G$ and respects~$\pi$. This amounts to conditions ${[RR_{r}]}$ and ${[RC_{\pi}]}$, while taking into account the refined notion of satisfation.
The notions of a differentiated quasi model and of a $\G$-strengthened  quasi model are defined as before, but using the new notion of satisfaction.

To be sure, in the case of $\GPK$, this means satisfying precisely the same semantic conditions laid down in Section~\ref{sec:proofsystem}, where we now take `$\tru{u}{\alpha}$' as abbreviating `$\ttt\in V(u,\alpha)$', and `$\fal{u}{\alpha}$' as abbreviating `$\fff\in V(u,\alpha)$'.
Whenever we need to distinguish between a semantic condition on a tuple $\langle w,\varphi\rangle$ as constraining a valuation~$V$ or a quasi valuation $QV\!$, we will use $\xval{w}{\varphi}$ for the former and  $\xvalQ{w}{\varphi}$ for the latter, where $\XX\in\{\TT,\FF\}$.

\subsection*{Step 2. Semantic proof of cut-admissibility}

The next step is to show that the existence of a countermodel in the form of a $\G$-strengthened differentiated quasi model implies the existence a countermodel in the form of a $\G$-legal model (which, in the case of $\GPK$ simply means an ordinary model).
For this purpose we 
define
an \textsl{instance} of a quasi model $\Q\M=\tup{\tup{W,R},QV}$ as any model of the form $\M=\tup{\tup{W,R'},V}$ such that
$\xvalQ{w}{\varphi}$ whenever $\xval{w}{\varphi}$, for every $\XX\in\{\TT,\FF\}$,
every $w\in W$ and every $\varphi\in{\cal L}$.  Note that a quasi model and its instances may have different accessibility relations.

Indeed, from \cite{lah:avr:Unified2013} we have:

\begin{theorem}[Corollary 5.48 of \cite{lah:avr:Unified2013}]
For every basic system $\G$, if every $\G$-strengthened differentiated quasi model  has a $\G$-legal instance, then $\G$ enjoys strong  cut-admissibility. 
\end{theorem}

And in particular:

\begin{corollary}
If every $\GPK$-strengthened differentiated quasi model satisfying the semantic conditions $[\TT\#]$ and $[\FF\#]$ from Section~\ref{sec:proofsystem} has an instance (that satisfies the same conditions), then $\GPK$ enjoys strong cut-admissibility.
\end{corollary}

\noindent 
Do note that the latter corollary uses the notion of satisfaction both for quasi-models and for ordinary models.
In what follows, the construction of appropriate instances is done by a recursive definition over the well-founded relation~$\wfr$ on the set of formulas, taken to be the smallest transitive relation satisfying the following:
$(i)$  if $\alpha$ is a proper subformula of $\beta$ then $\alpha\wfr\beta$;
$(ii)$ $\ineg\gamma\wfr\wsmile\gamma$ for every $\gamma\in{\cal L}$;
and $(iii)$ $\uneg\gamma\wfr\wfrown\gamma$ for every $\gamma\in{\cal L}$.
In what follows, $\alpha\wfreq\beta$ abbreviates $\alpha\wfr\beta\vee\alpha=\beta$.

\begin{lemma}
\label{quasinstance}
	Every  quasi model has an  instance.
\end{lemma}
\begin{proof}
	Let $\Q\M=\tup{\F,QV}$ be a quasi model based on a frame $\F=\tup{W,R}$.
We set up now an appropriate valuation $V:W\times\setS\rightarrow\set{\fff,\ttt}$.
For every world~$w$ and formula~$\varphi$, the valuation~$V$ is inductively defined (with respect to~$\wfr$) on~$\varphi$ as follows:
(R1) if $\truQ{w}{\varphi}$ fails for $QV\!$, we postulate $\fal{w}{\varphi}$ to be the case for~$V$;
(R2) if $\falQ{w}{\varphi}$ fails for $QV\!$, we postulate $\tru{w}{\varphi}$ to be the case for~$V$;
(R3) otherwise both $\truQ{w}{\varphi}$ and $\falQ{w}{\varphi}$ hold good for $QV\!$, and in this case we postulate $\tru{w}{\varphi}$ to be the case for~$V$ if one of the following holds:
\smallskip

\noindent
\begin{tabular}{rl}
  {(M1)} & $\varphi$ is a propositional variable or $\varphi$ is $\top$\\
  {(M2)} & $\varphi=\varphi_{1}\w\varphi_{2}$, and both $\tru{w}{\varphi_{1}}$ and $\tru{w}{\varphi_{2}}$\\
  {(M3)} & $\varphi=\varphi_{1}\vee\varphi_{2}$, and either $\tru{w}{\varphi_{1}}$ or $\tru{w}{\varphi_{2}}$\\
  {(M4)} & $\varphi=\ineg\psi$, and $\fal{v}{\psi}$ for some $v\in W$ such that $w R v$\\
  {(M5)} & $\varphi=\uneg\psi$, and $\fal{v}{\psi}$ for every $v\in W$ such that $w R v$\\
  {(M6)} & $\varphi=\wsmile\psi$, and either $\fal{w}{\psi}$ or $\fal{w}{\ineg\psi}$\\
  {(M7)} & $\varphi=\wfrown\psi$, and both $\fal{w}{\psi}$ and $\fal{w}{\uneg\psi}$\\
\end{tabular}\smallskip

\noindent
Otherwise, we postulate $\fal{w}{\varphi}$ to be the case for~$V\!$.
Obviously, $\xval{w}{\varphi}$ implies $\xvalQ{w}{\varphi}$ for every $w\in W\!$, every $\varphi\in{\cal L}$ and every $\XX\in\{\TT,\FF\}$.
It is a routine task to verify that $\tup{\F,V}$ is a model. 
We show here that the semantic conditions for $\ineg$ and $\wsmile$ hold:\\
\noindent
\textbf{[Case of $\ineg$]}
Let $\psi\in{\cal L}$.
Suppose first that $\fal{v}{\psi}$ is the case for some~$v\in W$ such that $wRv$.  Then $\falQ{v}{\psi}$.  Since $\Q\M$ is  $\GPK$-legal, then $\truQ{w}{\ineg\psi}$ is the case.  If, on the one hand, $\falQ{w}{\ineg\psi}$ fails, then we must have $\tru{w}{\ineg\psi}$, by (R2).  If, on the other hand, neither $\truQ{w}{\ineg\psi}$ nor $\falQ{w}{\ineg\psi}$ fail, we are in case (R3).  Since we have $\fal{v}{\psi}$ and $wRv$ we conclude by (M4) that $\tru{w}{\ineg\psi}$ must be the case.
Suppose now that $\tru{v}{\psi}$ is the case for every world~$v$ such that $wRv$.  Then we have $\truQ{v}{\psi}$ for every such world.  Since $\Q\M$ is $\GPK$-legal, it follows that $\falQ{w}{\ineg\psi}$ is the case.  If, on the one hand, $\truQ{w}{\ineg\psi}$ fails, then we must have $\fal{w}{\ineg\psi}$, by (R1). If, on the other hand, neither $\truQ{w}{\ineg\psi}$ nor $\falQ{w}{\ineg\psi}$ fail, we are in case (R3).  Since we have $\tru{v}{\psi}$ for every world $v$ such that $wRv$ we conclude that none of (M1)--(M7) applies, thus $\fal{w}{\ineg\psi}$ is to be the case.  \\
\textbf{[Case of }$\wsmile$\textbf{]}
Let $\psi\in\setS$.
Suppose first that either $\fal{w}{\psi}$ or $\fal{w}{\ineg\psi}$ are the case for some $w\in W$. Then either $\falQ{w}{\psi}$ or $\falQ{w}{\ineg\psi}$. Since $\Q\M$ is $\GPK$-legal, it follows that $\truQ{w}{\wsmile\psi}$. If, on the one hand, $\falQ{w}{\wsmile\psi}$ fails, then we must have $\tru{w}{\wsmile\psi}$, by (R2). If, on the other hand, neither $\truQ{w}{\wsmile\psi}$ nor $\falQ{w}{\wsmile\psi}$ fail, we are in case (R3) and we conclude by (M6) that $\tru{w}{\wsmile\psi}$ must be the case.
Suppose now that both $\tru{w}{\psi}$ and $\tru{w}{\ineg\psi}$ are the case for some $w\in W$. Then $\truQ{w}{\psi}$ and $\truQ{w}{\ineg\psi}$. Since $\Q\M$ is $\GPK$-legal, then $\falQ{w}{\wsmile\psi}$. If, on the one hand, $\truQ{w}{\wsmile\psi}$ fails, then we must have $\fal{w}{\wsmile\psi}$, by (R1). If, on the other hand, neither $\truQ{w}{\wsmile\psi}$ nor $\falQ{w}{\wsmile\psi}$ fail,
we are in case (R3) and $\fal{w}{\wsmile\psi}$ is to be the case because none of (M1)--(M7) applies.
\end{proof}

Since the class of all $\GPK$-legal quasi models contains the $\GPK$-strengthened differentiated quasi models, it follows that:

\begin{corollary}
\label{gpkcut}
	$\GPK$ enjoys strong cut-admissibility.
\end{corollary}

Given a basic system~$\G$ and a relation~$\sqsubset$ on its set of for\-mulas, we say that a derivation in $\G$ of a sequent $s$ from a set  $S$ of sequents  is a \textsl{$\sqsubset$-analytic derivation} if every formula~$\varphi$ that occurs in the derivation satisfies $(\varphi\sqsubset\psi)\vee(\varphi=\psi)$ for some~$\psi$ in $S\cup \{s\}$.
We then say that~$\G$ is \textsl{$\sqsubset$-analytic}~if whenever ~$s$ is derivable from  $S$  in~$\G$ there actually is some $\sqsubset$-analytic derivation of~$s$ from~$S$ in~$\G$.  
In case $\varphi\sqsubset\psi$ we may also say that~$\varphi$ is a \textsl{proper $\sqsubset$-subformula} of~$\psi$.

In view of Cor.~\ref{gpkcut}, the inner structure of the rules in $\GPK$ implies that:
\begin{corollary}\label{PKanalyticity}
$\GPK$ is $\wfr$-analytic.
\end{corollary}
\begin{proof}
By induction on the length of the derivation of $s$ from $S$ in $\GPK$:
In all rules except for $(cut)$, the premises 
include only formulas that already appear in the conclusion or are proper $\wfr$-subformulas of formulas that appear in the conclusion.
\end{proof}


Note that the $\wfr$-analyticity of PK immediately implies its \textit{decidability}: Given a finite set~$S$ of sequents and a sequent~$s$, we do not need to search for an arbitrary derivation of~$s$ from~$S$, but it suffices to search for a $\wfr$-analytic derivation. Clearly, the set of $\wfr$-analytic derivations whose set of premises is~$S$ is finite, and can be easily computed.

%% file: sec-extensions.tex
In this section we investigate several natural deductive extensions of $\GPK$.
Given a property~$X$ of binary relations, we call a frame $\tup{W,R}$ an \textsl{$X$~frame} if~$R$ enjoys~$X$. A (quasi) model $\tup{\F,V}$ is called an \textsl{$X$~$($quasi$)$ model} if $\F$ is an~$X$~frame. In addition, and similarly to what we did in the case of $\GPK$, for every proof system $Y$ we write $S\vdash_{Y}s$ if there is a derivation of~$s$ from~$S$ in~$Y$.

\subsection{Seriality}
\label{serialitySection}
Let $\GPKD$ be the system obtained
by augmenting $\GPK$ with the following rule:
{\small
$$[\D] \quad \dfrac{\Gamma\Ra\Delta}{\uneg\Delta\Ra\ineg\Gamma}$$
}%
This rule may be formulated as the basic rule:
$\tup{\Ra\;;\;\pi_{1}}\rs\Ra$.
Since its premise  is the empty sequent, the semantic condition it imposes (following~\cite{lah:avr:Unified2013}) is seriality: indeed, respecting $[\D]$ in a world $w$ of a model $\M$ based on a frame $\tup{W,R}$ means that if $\M,v\Vdash\; \Ra$ for every world $v$ such that $w R v$, then also $\M,w\Vdash\; \Ra$. Since the empty sequent is not satisfied at any world, this condition would hold iff for every world $w$ there exists a world $v$ such that $w R v$.  In addition, it is easy to see that every serial model satisfies this semantic condition.

As in Corollary \ref{soundnessAndCompletenessForSmiles}, we obtain a completeness theorem for $\GPKD$ with respect to serial models:
\begin{corollary}\label{soundnessAndCompletenessForSmilesatKD}
$\g\models_{\E_{\D}}\varphi$ iff $\g\vdash_{\GPKD}\varphi$ for every $\g\cup\set{\varphi}\suq\setS$, where $\E_{\D}$ is the class of serial frames.
\end{corollary}

It would now be straightforward to use rule $[\D]$ together with the latter result to see that (DT$\uneg$) and (DF$\ineg$) hold good in the logic $PKD$ (recall from \Cref{addingNEG} that (DT$\ineg$) and (DF$\uneg$) hold good for all extensions of $PK$).

Additionally, we may prove cut-admissibility also for the system $\GPKD$, going through serial quasi models.

\begin{lemma}
\label{Dinstance}
	Every  serial quasi model has a serial instance.
\end{lemma}
\begin{proof}
The proof is the same as the proof of Lemma \ref{quasinstance}.
Note that no property of the accessibility relation was assumed, and the constructed instance
has the same accessibility relation as the original quasi model.
\end{proof}


\begin{corollary}
	$\GPKD$ enjoys cut-admissibility and is $\wfr$-analytic.
\end{corollary}

\subsection{Reflexivity}

Let $\GPKT$ be the system obtained
by augmenting $\GPK$ with the following rules:
{\small
$$[\T_{1}] \quad \dfrac{\Gamma,\varphi\Ra\Delta}{\Gamma\Ra\ineg\varphi,\Delta}
\qquad\qquad\qquad [\T_{2}] \quad \dfrac{\Gamma\Ra\varphi,\Delta}{\Gamma,\uneg\varphi\Ra\Delta}$$}%
These rules may be formulated as the basic rules:
$\tup{p_{1}\Ra\;;\;\pi_{0}}\rs\Ra\ineg p_{1}$ and $\tup{\Ra p_{1}\;;\;\pi_{0}}\rs\uneg p_{1}\Ra\;$.  It should be clear that $\GPKT$ allows thus for the derivation of the consecutions representing $\llbracket\ineg$-implosion$\rrbracket$ and $\llbracket\uneg$-explosion$\rrbracket$.

Semantically, they impose reflexivity not on all models, but only on $\GPKT$-strengthened models. Indeed, since the underlying context relation is $\pi_{0}$, for every  model $\M=\tup{\F,V}$ based on a frame $\F=\tup{W,R}$ that respects $[\T_{1}]$ and $[\T_{2}]$,  and every world $w$, if $\M,w\vDash\varphi\Ra$  then  $\M,w\vDash\;\Ra\ineg\varphi$ and if $\M,w\vDash\;\Ra\varphi$ then  $\M,w\vDash\uneg\varphi\Ra$.
To put it otherwise,
if $\fal{w}{\varphi}$ then $\tru{w}{\ineg\varphi}$, and
if $\tru{w}{\varphi}$ then $\fal{w}{\uneg\varphi}$.
Clearly, every reflexive model satisfies the latter conditions. To show that every $\GPKT$-strengthened model that satisfies them is reflexive, consider an arbitrary  such model
 $\M=\tup{\tup{W,R},V}$. Then for every world $w\in W$ we have that for every formula $\varphi$, ($\tru{w}{\varphi}$ implies $\fal{w}{\uneg\varphi}$)
and ($\fal{w}{\varphi}$ implies $\tru{w}{\ineg\varphi}$), which in  $\GPKT$-strengthened models means precisely that $w R w$.
We therefore have that every reflexive model is  $\GPKT$-legal, and every $\GPKT$-strengthened model is reflexive.
We obtain thus a completeness theorem for $\GPKT$ with respect to reflexive models, relying on \Cref{OrisTheorem}:
\begin{corollary}\label{soundnessAndCompletenessForSmilesatKT}
$\g\models_{\E_{\T}}\varphi$ iff $\g\vdash_{\GPKT}\varphi$ for every $\g\cup\set{\varphi}\suq\setS$, where $\E_{\T}$ is the class of reflexive frames.
\end{corollary}

Such semantics for $\GPKT$ allows one to easily confirm that the full type diamond-minus connective~$\ineg$ fails (DM1.2\#), and that the full type box-minus connective~$\uneg$ fails (DM2.2\#).   Such failures transfer to the weaker logics $\GPKD$ and $\GPK$, of course.

Cut-admissibility for $\GPKT$ may be obtained using arguments similar to those used in proving Lemma~\ref{quasinstance}.
It follows thus that:

\begin{lemma}
	Every  reflexive $\GPKT$-strengthened quasi model has a reflexive instance.
\end{lemma}

\begin{corollary}
	$\GPKT$ enjoys cut-admissibility and is $\wfr$-analytic.
\end{corollary}

\subsection{Functionality}
In this section we address functional frames, that is, frames whose accessibility relations are {total functions}.
In every model $\tup{\tup{W,R},V}$ of a functional frame and world $w\in W$, we have
$\tru{w}{\ineg\varphi}$ iff $\tru{w}{\uneg\varphi}$. Hence $\ineg$ and $\uneg$ are indistinguishable.
Accordingly, here we consider a restricted language, without~$\uneg$.

Let $\GPKF$ be the system obtained from $\GPK$ by substituting~$\ineg$ for~$\uneg$ in rules $[{\Ra}\wfrown]$ and $[\wfrown{\Ra}]$, and replacing both rules $[{\Ra}\uneg]$ and $[\ineg{\Ra}]$ with the single rule:
{\small
$$[\FUNC]\quad\dfrac{\Gamma\Ra\Delta}{\ineg\Delta\Ra\ineg\Gamma}$$
}%
It is straightforward to see that rule $[\FUNC]$ may be formulated as the following basic rule:
$\tup{\Ra\;;\;\pi_{2}}\rs\Ra$, for $\pi_{2}=\set{\tup{q_{1}\Ra\;;\;\Ra\ineg q_{1}},\tup{\Ra q_{1}\;;\;\ineg q_{1}\Ra}}$.
Note that $\pi_{2}$ is obtained from $\pi_{1}$ by identifying $\ineg$ and $\uneg$.
The latter rule and context relation impose functionality on  \textit{differentiated} models. Indeed, respecting the basic rule $[\FUNC]$ corresponds to seriality, similarly to the case of the rule $[\D]$.
Additionally, the context relation~$\pi_{2}$ forces the accessibility relation to be a partial function:
Respecting~$\pi_{2}$ in a world $w$ of a model $\M=\tup{\tup{W,R},V}$ means that for every $v_{1},v_{2}\in W$ such that $w R v_{1}$ and $w R v_{2}$, and every formula $\varphi$,
$\tru{v_{1}}{\varphi}$ iff $\fal{w}{\ineg\varphi}$ iff
$\tru{v_{2}}{\varphi}$.  When $\M$ is differentiated, this implies that $v_{1}=v_{2}$.
Now, every functional model satisfies $[RR_{\FUNC}]$ and $[RC_{\pi_{2}}]$ and every differentiated model that satisfies  them is  functional.
We thus obtain a completeness result for $\GPKF$ with respect to functional models, relying on \Cref{OrisTheorem}:
\begin{corollary}
\label{soundAndCompleteFunctional}\label{soundnessAndCompletenessForSmilesatKF}
$\g\models_{\E_{\FUNC}}\varphi$ iff $\g\vdash_{\GPKF}\varphi$ for every $\g\cup\set{\varphi}\suq\setS$, where $\E_{\FUNC}$ is the class of functional frames.\footnote{We note  here that we do not have actual set-inclusion of~$\pi_{2}$ in~$\pi_{1}$. However, the language that we consider here identifies $\ineg$ and $\uneg$, and this suffices for our version of the general results from \cite{lah:avr:Unified2013}.}
\end{corollary}

It should be clear that $\GPKF$ extends $\GPKD$, but does not extend $\GPKT$.
Moreover, in contrast with what was the case for $\GPKT$, within the semantics for $\GPKF$  there are no longer countermodels  for (DM1.2$\ineg$)  or for (DM2.2$\uneg$).  

Going through quasi models we may prove cut-admissibility also for $\GPKF$.
However, unlike in previous cases, considering functional quasi models
will not suffice. Indeed, there exist $\GPKF$-strengthened differentiated quasi models  whose accessibility relation is not a total function.
Nonetheless, it can be verified that every $\GPKF$-legal quasi model that is  based on a frame $\F=\tup{W,R}$ is  serial, and for every $w,v\in W$ such that $w R v$ we have, for every $\varphi\in\setS$, both
($\falQ{v}{\varphi}$ implies $\truQ{w}{\ineg\varphi}$) and ($\truQ{v}{\varphi}$ implies $\falQ{w}{\ineg\varphi}$).
Thus although the accessibility relation in $\GPKF$-legal~quasi models may not be a total function, we are still able to extract a functional model from it:

\begin{lemma}\label{func-instance}
	Every  $\GPKF$-legal~quasi model has a functional instance.
\end{lemma}
\begin{proof}
Let $\Q\M=\tup{\F,QV}$ be an $\GPKF$-legal~quasi model based on a frame $\tup{W,R}$. Since $\Q\M$ is $\GPKF$-legal, we have in particular that~$R$ is serial. Therefore, there exists some $R':W\rightarrow W$ such that $R'\suq R$. Let $\F'=\tup{W,R'}$. We define an appropriate valuation  $V:W\times\setS\rightarrow\set{\fff,\ttt}$ as in Lemma \ref{quasinstance}, while disregarding {(M5)}, and using the following two instructions in place of {(M4)} and {(M7)}:
\smallskip

\noindent
\begin{tabular}{rl}
  {(M4$^\star$)} & $\varphi=\ineg\psi$, and $\fal{R'(w)}{\psi}$\\
  {(M7$^\star$)} & $\varphi=\wfrown\psi$, and $\fal{w}{\varphi}$ and $\fal{w}{\ineg\varphi}$\\
\end{tabular}\smallskip

\noindent
The proof then carries on in a similar fashion to the proof of Lemma \ref{quasinstance}.
\end{proof}

\begin{corollary}
	$\GPKF$ enjoys cut-admissibility and is $\wfr'$-analytic, where $\wfr'$ is the restriction of $\wfr$ to the $\uneg$-free fragment of $\setS$, with an additional clause according to which $\ineg\varphi\wfr \wfrown\varphi$.
\end{corollary}

We include a brief comment concerning a decision procedure for this logic.
It is easy to see that~$\ineg$ and~$\uneg$ may be defined using the customary presentation of the modal logic~$\K$ by $\ineg\varphi:={\sim}\Box\varphi$ and $\uneg\varphi:=\Box{\sim}\varphi$. When considering only functional frames (like in $\GPKF$), we get a translation to $\GKF$ --- the ordinary modal logic of functional Kripke models. For the $\wsmile\wfrown$-free fragment of this logic, we may apply the general reduction to SAT proposed in \cite{lahavZoharSatBased}, which in particular means that the derivability problem for it is in co-NP.

\subsection{Symmetry}
\label{sec-symmetry}
Let $\GPKB$ be the system obtained from $\GPK$ by replacing ${[{\ineg}{\Ra}]}$ and ${[{\Ra}{\ineg}]}$  with the following rules:
{\small
$$[\B_{1}] \quad \dfrac{\Gamma,\ineg\Gamma',\varphi\Ra\Delta,\uneg\Delta'}{\uneg\Delta,\Delta'\Ra\uneg\varphi,\ineg\Gamma,\Gamma'}
\qquad\qquad\qquad [\B_{2}] \quad \dfrac{\Gamma,\ineg\Gamma'\Ra\varphi,\Delta,\uneg\Delta'}{\uneg\Delta,\Delta',\ineg\varphi\Ra\ineg\Gamma,\Gamma'}$$}%

\noindent
These correspond to the following basic rules:
$\tup{ p_{1}\Ra\;;\;\pi_{3}}\rs\Ra\uneg p_{1}$ and  $\tup{\Ra p_{1}\;;\;\pi_{3}}\rs\ineg p_{1}\Ra$,
for the context relation $$\pi_{3}=\pi_{1}\cup\set{\tup{\ineg q_{1}\Ra\;;\;\Ra q_{1}},\tup{\Ra\uneg q_{1}\;;\; q_{1}\Ra}}.$$
This relation satisfies the following property: $s_1\,\pi_{3}\,s_2$ iff $\overline{s_2}\,\pi_{3}\,\overline{s_1}$, where $(\overline{\Ra\varphi})$ denotes $(\varphi\Ra)$ and $(\overline{\varphi\Ra})$ denotes $(\Ra\varphi)$.
By Proposition 4.28 of \cite{lah:avr:Unified2013}, 
$\GPKB$-strengthened models are symmetric.
In addition, every symmetric model respects these rules, as well as the context relation~$\pi_{3}$: $[RR_{\B_{1}}]$ and $[RR_{\B_{2}}]$ are shown similarly to the case of $\GPK$. As for $[RC_{\pi_{3}}]$, suppose $w R u$ in some symmetric model $\M$ based on a frame $\F=\tup{W,R}$. If $V,u\Vdash \ineg\varphi\Ra$ then $V(u,\ineg\varphi)=0$. Since $\M$ is  symmetric, we have $u R w$ as well, and since it is a model, it follows that $V(w,\varphi)=1$, which means that $V,w\Vdash\;\Ra\varphi$. Similarly, if $V,u\Vdash\;\Ra\uneg\varphi$ then $V,w\Vdash\varphi\Ra\;$.
Based on \Cref{OrisTheorem}, we see that:
\begin{corollary}\label{soundnessAndCompletenessForSmilesatKB}
$\g\models_{\E_{\B}}\varphi$ iff $\g\vdash_{\GPKB}\varphi$ for every $\g\cup\set{\varphi}\suq\setS$,
where $\E_{\B}$ is the class of symmetric frames.
\end{corollary}

Symmetric frames are also relevant from the viewpoint of sub-classical properties of negation.  They validate, for instance, the consecutions $\ineg\ineg p\models p$ and {$p\models \uneg\uneg p$}.
\Cref{frameProperties} (\Cref{sec:coda}), collects these and also many other consecutions representing forms of De Morgan rules that are validated by symmetric frames.
It is also worth noting that \Cref{framePropertiesC} (\Cref{sec:coda}), contains some global inference rules that are made valid by the demand of symmetry.  To check rule $\frac{\ineg\varphi\;\models\;\psi}{\ineg\psi\;\models\;\varphi}$, for instance, assume $\ineg\varphi\;\models\;\psi$ and suppose that ${\cal M}, w\Vdash \ineg\psi$ at a world~$w$ of a model~$\cal M$ of some symmetric frame.  By [S$\ineg$] we know that there must be some world~$v$ such that $wRv$ and ${\cal M}, v\not\Vdash \psi$. From the assumption it follows that ${\cal M}, v\not\Vdash \ineg\varphi$. Given that~$R$ is symmetric, we know that $vRw$, thus we can conclude that ${\cal M}, w\Vdash \varphi$, as we intended to.
Rule $\frac{\varphi\;\models\;\uneg\psi}{\psi\;\models\;\uneg\varphi}$ may be checked in an analogous way.
Paraconsistent logics based on symmetric (and reflexive) frames are also studied in~\cite{avr:zam:AiML16}, a paper that investigates in detail a conservative extension of the corresponding logic, obtained by the addition of a classical implication (but without primitive~$\uneg$ and~$\wfrown$), and offers for this logic a sequent system for which cut is \textit{not} eliminable.

Quasi models for $\GPKB$ are not necessarily symmetric, making it harder to convert them into instances in the form of symmetric models.  Hence, the particular question of cut-admissibility for our system $\GPKB$ goes beyond the reach of our approach, and is left open as a matter for further research. 
However, it can still be shown that $\GPKB$ is $\wfr$-analytic.  Unlike we did for the above systems, the latter result is not to be obtained as a corollary of cut-admissibility, but will be shown directly, using a similar technique.

As we did before for cut-admissibility,  $\wfr$-analyticity may also be shown in two steps: \textit{First}, we present an adequate semantics for analytic derivations in $\GPKB$;
\textit{second}, we show that a countermodel in the new semantics entails the existence of a countermodel in the form of a Kripke model, as defined in \Cref{sec:proofsystem}.
Rather than using quasi valuations, for this purpose we use `partial valuations'. Models based on partial valuations are very similar to the usual Kripke models. The only difference is that the underlying interpretation is partial --- that is, not defined over all formulas of the language. 
The exact same semantic conditions read off the derivation rules is imposed on partial valuations. Concretely, given a frame $\F=\tup{W,R}$, a \textsl{partial valuation} over it is a partial function $PV$ from $W\times\setS$ to $\set{\fff,\ttt}$ satisfying precisely the same semantic conditions laid down in \Cref{sec:proofsystem}, where each condition is restated so as to apply only to formulas that are assigned a  truth value. We denote the set of formulas that are assigned a value by a partial valuation $PV$ in a world $w$ of $\F$ by $\dom{PV,w}$. For example, {[$\TT\ineg$]} now reads as:  if $\fal{v}{\varphi}$ for some $v\in W$ such that $wRv$ and $\ineg\varphi\in\dom{v,w}$, then $\tru{w}{\ineg\varphi}$.
 
A \textsl{partial model} is a structure $PM=\tup{\F,PV}$, where $PV$ is a partial valuation over $\F$. The notions of a differentiated partial model and of a $\GPKB$-strengthened partial model are defined as before.
When $\dom{PV,w}=X$  for every world $w$ of $\F$, we call $PM$ a \textsl{partial $X$-model}.
	
	Now, Proposition 4.28 of \cite{lah:avr:Unified2013} is generalized there (Proposition 5.21) to cover partial $X$-models, for any set $X$ of formulas, and thus symmetry is imposed also on partial models that respect $[RC_{\pi_{3}}]$.

\begin{theorem}[Corollary 5.48 of \cite{lah:avr:Unified2013}]
Let $\G$ be a basic system. If, for every finite set $X$ that is closed under $\wfr$, it holds that every differentiated $\G$-strengthened partial $X$-model can be extended to a $\G$-legal model, then $\G$ is $\wfr$-analytic.
\end{theorem}

And in particular:

\begin{corollary}
If, for every finite  set $X$ that is closed under $\wfr$, it holds that every differentiated $\GPKB$-strengthened partial $X$-model satisfying the semantic conditions $[\TT\#]$ and $[\FF\#]$ from Section~\ref{sec:proofsystem} can be extended to a symmetric model, then $\GPKB$ is $\wfr$-analytic.
\end{corollary}

To establish $\wfr$-analyticity, it suffices to show that:

\begin{lemma}
\label{symmExtend}
Let $X\suq\setS$ be closed under~$\wfr$.
Every symmetric $\GPKB$-strength\-ened partial $X$-model can be extended to a symmetric  model.
\end{lemma}

\begin{proof}
The proof is very similar to the proof of \Cref{quasinstance}. There, we provided a recursive procedure to eliminate the value $\set{\ttt,\fff}$ from a model. Here, instead of having formulas that are assigned $\set{\ttt,\fff}$, we have formulas that do not receive any assignment (as the model is \textit{partial}). We then treat such unassigned formulas as formulas that have been assigned the value $\set{\ttt,\fff}$. More concretely, the condition (R3) from the proof of \Cref{quasinstance} is replaced from ``both $\truQ{w}{\varphi}$ and $\falQ{w}{\varphi}$ hold good" to ``neither $\truQ{w}{\varphi}$ nor $\falQ{w}{\varphi}$ hold good". Then the proof carries on similarly, following the same procedure induced by (M1)--(M7).
\end{proof}

Thus, using Corollary 5.44 of \cite{lah:avr:Unified2013}, we conclude that:
\begin{corollary}
$\GPKB$ is $\wfr$-analytic.
\end{corollary}

Note that $\wfr$-analyticity for $\GPKB$ was established directly, not through cut-admissibility. For this reason, the inner structure of the rules (that does not enjoy a local $\wfr$-subformula property) did not matter here, but only the fact that partial models can be appropriately extended. 


\subsection{Transitivity}
\label{transitivitySection}
Let $\GPKFour$ be the system obtained from $\GPK$ 
by replacing ${[{\ineg}{\Ra}]}$ and ${[{\Ra}{\ineg}]}$  with the following rules:
$$[\KFour_{1}]
\quad
\dfrac{\uneg\Gamma,\Gamma',\varphi\Ra\ineg\Delta,\Delta'}{\uneg\Gamma,\uneg\Delta'\Ra\uneg\varphi,\ineg\Delta,\ineg\Gamma'}
\qquad\qquad\qquad 
[\KFour_{2}]
\quad
\dfrac{\uneg\Gamma,\Gamma'\Ra\varphi,\ineg\Delta,\Delta'}{\uneg\Gamma,\uneg\Delta',\ineg\varphi\Ra\ineg\Delta,\ineg\Gamma'}
$$
These correspond to the following basic rules: 
$\tup{ p_{1}\Ra\;;\;\pi_{4}}\rs\Ra\uneg p_{1}$ and  $\tup{\Ra p_{1}\;;\;\pi_{4}}\rs\ineg p_{1}\Ra$,
for the context relation $$\pi_{4}=\pi_{1}\cup\set{\tup{\uneg q_{1}\Ra\;;\;\uneg q_{1}\Ra},\tup{\Ra\ineg q_{1}\;;\;\Ra\ineg q_{1}}}.$$

For this relation, we have $\pi_{4}=\pi_{4}\circ\pi_{4}$.
By Proposition 4.28 of \cite{lah:avr:Unified2013}, the semantic condition imposed on $\GPKFour$-strengthened models is transitivity of the accessibility relation. In addition, every transitive model respects rules $[\KFour_{1}]$ and $[\KFour_{2}]$, and also respects the context relation $\pi_{4}$. 
For example, if $w R u$ and $V,u\Vdash\uneg\varphi\Ra$, then $V(u,\uneg\varphi)=0$. This means that there is some $v$ such that $u R v$ and $V(v,\varphi)=1$. By transitivity, we have also $w R v$ and therefore $V(w,\uneg\varphi)=0$, which means that $V,w\Vdash\uneg\varphi\Ra$.
Then, \Cref{OrisTheorem} gives us:

\begin{corollary}\label{soundnessAndCompletenessForSmilesatK4}
$\g\models_{\E_{\KFour}}\varphi$ iff $\g\vdash_{\GPKFour}\varphi$ for every $\g\cup\set{\varphi}\suq\setS$,
where $\E_{\KFour}$ is the class of transitive frames.
\end{corollary}

From the viewpoint of sub-classical properties of negation, some important consecutions validated by transitive frames are collected in \Cref{framePropertiesD} (\Cref{sec:coda}).  And transitivity also has a role to play at \Cref{framePropertiesB,morenegations} (\Cref{sec:coda}), in the presence of appropriate adjustment connectives.

As for cut-admissibility, Prop.~5.21 of \cite{lah:avr:Unified2013} guarantees that strengthened quasi-models that respect $[\KFour_{1}]$ and $[\KFour_{2}]$ and~$\pi_{4}$ are transitive as well. In addition, it is easy to verify that transitive quasi models respect them. Therefore, arguing about cut-admissibility reduces again to finding an instance to every transitive quasi model. Now, since the properties of the accessibility relation bore no effect on this procedure, as was demonstrated in the proof of \Cref{quasinstance}, we obtain:

\begin{corollary}
$\GPKFour$ enjoys cut-admissibility and is $\wfr$-analytic.
\end{corollary}

\subsection{Combining frame properties}
\label{serialityTransitivitySection}

In some cases it is also possible to apply our machinery to provide useful sequent systems for logics whose semantics combine more than one of the properties studied in isolation in the previous sections.
We present two illustrations of that in the present section.  First we consider the system $\GPKFourD$ obtained from $\GPKFour$ by augmenting the latter with the following rule:
$$[\D_{4}]
\quad
\dfrac{\uneg\Gamma',\Gamma\Ra\Delta,\ineg\Delta'}{\uneg\Gamma',\uneg\Delta\Ra\ineg\Gamma,\ineg\Delta'}
\qquad\qquad\qquad 
$$
It corresponds to the following basic rule: 
$\tup{ \Ra\;;\;\pi_{4}}\rs\Ra$,
for the same context relation~$\pi_{4}$ defined in \Cref{transitivitySection}.
Similarly to what was shown in \Cref{serialitySection}, one may now show that every model (and in particular every transitive model) that respects this rule is serial, and every serial model respects it. The same holds for quasi models.
Together with what we have just seen regarding transitive (quasi) models, we note that every model that is both transitive and serial respects $[\D_{4}]$ and $\pi_{4}$, and every $\GPKFourD$-strengthened model that respects them is transitive and serial. Proceeding exactly as before, we now obtain that:

\begin{corollary}
$\g\models_{\E_{\GPKFourD}}\varphi$ iff $\g\vdash_{\GPKFourD}\varphi$ for every $\g\cup\set{\varphi}\suq\setS$,
where $\E_{\GPKFourD}$ is the class of serial transitive frames.
In addition,
$\GPKFourD$ enjoys cut-admissibility and is $\wfr$-analytic.
\end{corollary}

Frames which are at once serial and transitive have a role to play validating some consecutions concerning qualified forms of explosion and implosion, found at \Cref{framePropertiesD} (\Cref{sec:coda}).

Next, let $\GPKDB$ be the system that augments the system $\GPKB$ with the following rule:
$$[\D_{B}] \quad \dfrac{\ineg\Gamma',\Gamma\Ra\Delta,\uneg\Delta'}{\Delta',\uneg\Delta\Ra\ineg\Gamma,\Gamma'}$$ 
%
This rule corresponds to the following basic rule:
$\tup{ \Ra\;;\;\pi_{3}}\rs\Ra$, for the relation $\pi_{3}$ defined in \Cref{sec-symmetry}.

Similarly to what was shown in \Cref{serialitySection}, one may now show that every model (and in particular every symmetric model) that respects this rule is serial, and every serial model respects it. The same holds for partial models.
Together with what we have seen in \Cref{sec-symmetry} regarding symmetric (quasi) models, we have that every model that is both symmetric and serial respects $[\D_{B}]$ and $\pi_{3}$, and every $\GPKDB$-strengthened model that respects them is symmetric and serial. Proceeding exactly as before, we obtain that:

\begin{corollary}
$\g\models_{\E_{\GPKDB}}\varphi$ iff $\g\vdash_{\GPKDB}\varphi$ for every $\g\cup\set{\varphi}\suq\setS$, 
where $\E_{\GPKDB}$ is the class of serial symmetric frames.
In addition,
$\GPKDB$  is $\wfr$-analytic.
\end{corollary}

Frames which are at once serial and symmetric play a role validating some mixed double negation consecutions found at \Cref{morenegations} (\Cref{sec:coda}).  We leave it to the reader to check the special form of global contraposition that is validated by the latter frames, found at the last row of \Cref{framePropertiesC} (\Cref{sec:coda}).

%% file: sec-definability.tex
In this section we investigate definability of classical negation in the modal logics studied in this paper.
Given a set~$C$ of connectives and a logic~${\bf L}$, we denote by $\Lano{{\bf L}}{C}$ the $C$-free fragment of~${\bf L}$, that is, the restriction of~${\bf L}$ to the language without the connectives in~$C$.
\bigskip

\begin{theorem}~
\begin{enumerate}
	\item Classical negation is definable in the logics: \\
	$\Lano{PKT}{\set{{\ineg},\wsmile}}$, $\Lano{PKT}{\set{{\uneg},\wfrown}}$, 
	$PKD$,  $PKF$, $PKD4$ and $PKDB$.
	\item Classical negation is {\em not\!} definable in the logics: \\
	 $PK$, $PKB$, $PK4$, $\Lano{PKT}{\set{\wsmile,\wfrown}}$, $\Lano{PKDB}{\set{\wsmile,\wfrown}}$, \\
	 $\Lano{PKD}{\set{\wsmile}}$, $\Lano{PKD}{\set{\wfrown}}$, $\Lano{PKF}{\set{\wsmile}}$,  $\Lano{PKF}{\set{\wfrown}}$, \\
	  $\Lano{PKD4}{\set{\wfrown}}$ and $\Lano{PKD4}{\set{\wsmile}}$.
\end{enumerate}
\end{theorem}

\begin{proof}~\\
(1)
For $\Lano{PKT}{\set{{\ineg},\wsmile}}$ we set ${\sim}\varphi:=\uneg\varphi\vee\wfrown\varphi$, for $\Lano{PKT}{\set{{\uneg},\wfrown}}$ we set ${\sim}\varphi:=\ineg\varphi\wedge\wsmile\varphi$, and for $PKD$ and $PKF$ we may set ${\sim}\varphi:=(\uneg\varphi\wedge\wsmile\varphi)\vee\wfrown\varphi$ 
or, dually, set ${\sim}\varphi:=(\ineg\varphi\vee\wfrown\varphi)\wedge\wsmile\varphi$.
It is easy to see that $\Ra\varphi,{\sim}\varphi$ and $\varphi,{\sim}\varphi\Ra$ are derivable in each system
for the defined connective~${\sim}$.
Using cut, one obtains the usual sequent rules for classical negation.
\Cref{PKTuDef} exhibits the derivations for $\Lano{PKT}{\set{\wsmile,{\ineg}}}$ (the derivations for  $\Lano{PKT}{\set{\wfrown,{\uneg}}}$ are analogous), and
\Cref{PKDDef} provides the derivations for $PKD$ for the first definition above (we leave the second as an exercise for the reader).
Given that $\GPKF$ is a deductive extension of $\GPKD$, the derivations in \Cref{PKDDef} are also good for $PKF$. The same holds for $\GPKD4$ and for $\GPKDB$, as the rule $[\D]$ is a particular instance of the rule $[\D_4]$ 
and also a particular instance of the rule $[\D_B]$.

\begin{figure}[hb]
\centering
  \begin{tabular}{cc}
    $ \dfrac{\dfrac{\varphi\Ra\varphi}{\varphi,\uneg\varphi\Ra}[\T_{2}]\;\;\;\;\dfrac{\varphi\Ra\varphi,\uneg\varphi}{\varphi,\wfrown\varphi\Ra}[{\wfrown}{\Ra}]}{\varphi,\uneg\varphi\vee\wfrown\varphi\Ra}[{\vee}{\Ra}]$
&
        $ \dfrac{\dfrac{{\varphi\Ra\varphi}\;\;\;\;{\uneg\varphi\Ra\uneg\varphi}}{\Ra\varphi,\uneg\varphi,\wfrown\varphi}[{\Ra}{\wfrown}]}{\Ra\varphi,\uneg\varphi\vee\wfrown\varphi}[{\Ra}{\lor}] $
  \end{tabular}
  \caption{Definability of negation in $\Lano{PKT}{\set{\protect\wsmile,{\ineg}}}$}
 \label{PKTuDef}
\end{figure}

\begin{figure}
\centering
$\dfrac{
	\dfrac{
		\dfrac{
			\ddfrac{}{\varphi\Ra\varphi} \;\;\;\;
			\dfrac{
				\varphi\Ra\varphi
			}{
				\uneg\varphi\Ra\ineg\varphi
			}[\D]
			}{
				\varphi,\uneg\varphi,\wsmile\varphi\Ra
			}[{\wsmile}{\Ra}]
		}{
			\varphi,\uneg\varphi\wedge\wsmile\varphi\Ra
		}[{\land}{\Ra}] \;\;\;\;
		\dfrac{\varphi\Ra\varphi,\uneg\varphi}{
			\varphi,\wfrown\varphi\Ra
		}[{\wfrown}{\Ra}]
	}{
		\varphi,(\uneg\varphi\wedge\wsmile\varphi)\vee\wfrown\varphi\Ra
	}[{\lor}{\Ra}]$

\bigskip		
    $\dfrac{\dfrac{\dfrac{\varphi\Ra\varphi \;\;\;\; \uneg\varphi\Ra\uneg\varphi}{\Ra\varphi,\uneg\varphi,\wfrown\varphi}[{\Ra}{\wfrown}] \;\;\;\; \dfrac{\varphi,\ineg\varphi\Ra\varphi,\wfrown\varphi}{\Ra\varphi,\wsmile\varphi,\wfrown\varphi}[{\Ra}{\wsmile}]}{\Ra\varphi,(\uneg\varphi\wedge\wsmile\varphi),\wfrown\varphi}[{\Ra}{\land}]}{\Ra\varphi,(\uneg\varphi\wedge\wsmile\varphi)\vee\wfrown\varphi}[{\Ra}{\lor}]$
    \caption{Definability of negation in ${PKD}$}
  \label{PKDDef}
\end{figure}

\medskip

\noindent
(2) 
		Let $X$ be one of the logics listed in the statement, and suppose for the sake of contradiction that classical negation~$\sim$ is definable in~$X$. Let $p\in\setP$ and~let $\varphi$ be ${\sim}(p)$. Then both $\Ra\varphi,p$ and $p,\varphi\Ra$ are valid in~$X$.
		Consider a set $W$ that consists of two worlds, $w$ and $v$, and a valuation $V$ such that $V(w,q)=1$ and $V(v,q)=0$ for every atomic formula~$q$ (includ\-ing~$p$).
		Now, for each relation $R_X$ on $W$,  consider the model $\M_X=\tup{\tup{W,R_X},V}$.
		If $\M_X$ belongs to the class of models that semantically characterize $X$, then we must have  that $\M_X,w\Vdash\varphi,p\Ra$ and $\M_X,v\Vdash\;\Ra p,\varphi$. Since in $\M_X$ we have $\tru{w}{p}$ and $\fal{v}{p}$, we must then have $\fal{w}{\varphi}$ and $\tru{v}{\varphi}$.
		We show that this is impossible, by structural induction on $\varphi$. More precisely, we claim that if $\fal{w}{\varphi}$ then $\fal{v}{\varphi}$.
		To show this, we consider the given values for~$X$, and define the accessibility relation $R_X$ in each case.
		We divide the possible values for $X$ into four cases:

\begin{enumerate}[(A)]
			\item For $X{\in}\set{PK,PKB,PK4}$ define $R_{X}=\varnothing$. Note that for these three logics,
			$M_X$ belongs to the appropriate class of models.
			 If $\varphi\in\setP$, then $\tru{w}{\varphi}$ by the definition of the valuation, hence the claim trivially holds. If $\varphi=\varphi_{1}\wedge\varphi_{2}$ for some $\varphi_{1},\varphi_{2}$ and $\fal{w}{\varphi}$, then by [$\TT\land$] we must have that either $\fal{w}{\varphi_{1}}$ or $\fal{w}{\varphi_{2}}$. By the induction hypothesis, we have that either $\fal{v}{\varphi_{1}}$ or $\fal{v}{\varphi_{2}}$, hence $\fal{v}{\varphi}$ by [$\FF\land$]. If $\varphi=\varphi_{1}\vee\varphi_{2}$ for some $\varphi_{1},\varphi_{2}$ then this is shown similarly.
			Now, since $R_{X}=\varnothing$, for every formula~$\psi$ we have (a) $\fal{w}{\ineg\psi}$ and (b) $\fal{v}{\ineg\psi}$. In view of (b), the claim is true if $\varphi=\ineg\psi$ for some~$\psi$.  By [$\TT\wsmile$], from (a) we conclude that $\tru{w}{\wsmile\psi}$. Thus, the claim also holds if $\varphi=\wsmile\psi$, for some~$\psi$.
			Similarly, for every formula~$\psi$ we have (c) $\tru{w}{\uneg\psi}$ and (d) $\tru{v}{\uneg\psi}$. In view of (c), the claim is true if $\varphi=\uneg\psi$ for some~$\psi$.  By [$\FF\wfrown$], from (d) we conclude $\fal{v}{\wfrown\psi}$. Thus, the claim also holds if $\varphi=\wfrown\psi$, for some~$\psi$.

			\item\label{reflexiveCase} For $X=\Lano{PKT}{\set{\wsmile,\wfrown}}$ or $X=\Lano{PKDB}{\set{\wsmile,\wfrown}}$ define $R_{X}=W\times W$. First note that since $R_{X}$ is reflexive and also symmetric, $M_X$ belongs to the appropriate class of models. If $\varphi$ is a propositional variable, a conjunction, or a disjunction, then the proof is analogous to the previous case. If $\varphi=\ineg\psi$ for some~$\psi$ and $\fal{w}{\varphi}$, then $\tru{w}{\psi}$ and $\tru{v}{\psi}$ by [$\TT\ineg$] and the definition of~$R_{X}$, which implies by [$\FF\ineg$] and the definition of~$R_{X}$ that $\fal{v}{\varphi}$. If $\varphi=\uneg\psi$ for some~$\psi$ and $\fal{w}{\varphi}$, then by [$\TT\uneg$] and the definition of~$R_{X}$ it follows that either $\tru{w}{\psi}$ or $\tru{v}{\psi}$. Either way we conclude by [$\FF\uneg$] and the definition of~$R_{X}$ that $\fal{v}{\varphi}$.

			\item For $X{\in}\set{\Lano{PKD}{\set{\wsmile}},\Lano{PKF}{\set{\wsmile}},\Lano{PKD4}{\set{\wsmile}}}$ set $R_{X}{=}\set{\tup{w,v}\!,\tup{v,v}}$.
			Note that since $R_{X}$ is a total function, 
 $\ineg$ and $\uneg$ are indistinguishable, hence we may choose to consider~$\uneg$ instead of~$\ineg$. Moreover, since $R_{X}$ is also transitive, $\M_X$ belongs to the appropriate class of models.  The cases where $\varphi$ is atomic, a conjunction, or a disjunction are immediate. If $\varphi=\uneg\psi$  for some~$\psi$ and $\fal{w}{\varphi}$, then we must have $\tru{v}{\psi}$ by [$\TT\uneg$], which implies by [$\FF\uneg$] that $\fal{v}{\varphi}$. If $\varphi=\wfrown\psi$ for some~$\psi$, then $\fal{v}{\varphi}$ must hold good: indeed, if on the one hand $\tru{w}{\uneg\psi}$ then $\fal{v}{\psi}$ by $[\FF\uneg]$, and hence $\tru{v}{\uneg\psi}$ by $[\TT\uneg]$, which implies by [$\FF\wfrown$] that $\fal{v}{\varphi}$; if on the other hand $\fal{w}{\uneg\psi}$ then $\tru{v}{\psi}$ by [$\TT\uneg$], and hence again $\fal{v}{\varphi}$ follows by [$\FF\wfrown$].

   			\item For $X{\in}\!\set{\Lano{PKD}{\set{\wfrown}},\Lano{PKF}{\set{\wfrown}},\Lano{PKD4}{\set{\wfrown}}}$ set $R_{X}{=}\set{\tup{w,w}\!,\tup{v,w}}$.
			For this case also,
			note that since $R_{X}$ is both a total function and a transitive relation, 
			$\M_X$ belongs to the appropriate class of models. 
			The proof then proceeds as in item (C), \textit{mutatis mutandis}.
\qed
\end{enumerate}
\renewcommand{\qed}{}
\qed
\end{proof}

Concerning the second part of the preceding proof, it is worth remarking that in the cases in which classical negation turned out to be definable some adjustment operator was always used for that purpose. 
Moreover, the proof of case (\ref{reflexiveCase}) actually shows that whenever all derivations in a given logic are sound with respect to any class of models that includes this model, classical negation is not definable without the help of $\wsmile$ or $\wfrown$, as the same argument applies. In particular, this includes all normal (negative) modal logics up to (negative) $S5$, and indeed all the logics in which $\wsmile\varphi$ is not equivalent to~$\top$ (logics which are $\smile$-paraconsistent) and $\wfrown\varphi$ is not equivalent to~$\bot$ (logics which are $\frown$-paracomplete).

Our next and final section will revisit the introductory comments of the paper in the light of what we have learned so far.

%% file: sec-coda.tex

\paragraph{Denying instead of affirming}
In contrast to the usual `positive modalities' of normal modal logics, which are monotone with respect to the underlying notion of consequence, we have devoted this paper to antitone connectives known as `negative modalities' --- specifically, to full type box-minus and full-type diamond-minus connectives.


Be they monotone or antitone on each of their arguments, the connectives of normal modal logics are always congruential: they treat equivalent formulas as synonymous.  The phenomenon seems to be an exception rather than the rule if many-valued logics with non-classical negations are involved.
For instance, Kleene's 3-valued logic fails to be congruential, as $p\land\neg p$ is equivalent to $q\land\neg q$, but their respective negations, $\neg(p\land\neg p)$ and $\neg(q\land\neg q)$, are not equivalent (where $\neg$ is Kleene's negation).
Also, the earliest paraconsistent logic in the literature (cf.~\cite{jas:48}) fails to be congruential, in spite of having been defined in terms of a double translation into a fragment of the modal logic $S5$, and such failure remained undetected for decades (cf.~\cite{jmar:04g}).  The same holds for the other early paraconsistent logics developed later on, containing additional `strong negations' that live in the vicinity of classical negation (cf.~\cite{Nel:NaSoCiCS:59,daC:comptes:63}).
Of course, there are important `non-exceptions': intuitionistic logic and other intermediate logics constitute paracomplete logics with the replacement property.  
For another example of the latter kind, perhaps more to the point, consider the four-valued logic of FDE, whose semantics may be formulated having as truth-values $\{\mathbf{t},\mathbf{b},\mathbf{n},\mathbf{f}\}$, where $\{\mathbf{t},\mathbf{b}\}$ are designated, the reflexive transitive closure of the order~$\leq$ such that $\mathbf{f}\leq\mathbf{n}$, $\mathbf{f}\leq\mathbf{b}$, $\mathbf{n}\leq\mathbf{t}$ and $\mathbf{b}\leq\mathbf{t}$ may be used to define $\land$ and $\lor$, respectively, as its meet and its join, while negation is defined by setting $\neg\langle\mathbf{t},\mathbf{b},\mathbf{n},\mathbf{f}\rangle:=\langle\mathbf{f},\mathbf{b},\mathbf{n},\mathbf{t}\rangle$.
It is not hard to see that this logic is congruential and by defining the operators
$\wsmile\langle\mathbf{t},\mathbf{b},\mathbf{n},\mathbf{f}\rangle:=\langle\mathbf{t},\mathbf{n},\mathbf{b},\mathbf{t}\rangle$ and
$\wfrown\langle\mathbf{t},\mathbf{b},\mathbf{n},\mathbf{f}\rangle:=\langle\mathbf{f},\mathbf{n},\mathbf{b},\mathbf{f}\rangle$
it gets conservatively extended into another congruential logic that deductively extends our logic $PKF$ (but does not deductively extend $PKT$), if we read~$\ineg$ as~$\neg$.  It is worth noting that the latter logic is equivalent (through a definitional translation) to the expansion of FDE by the addition of the operator ${\sim}\langle\mathbf{t},\mathbf{b},\mathbf{n},\mathbf{f}\rangle:=\langle\mathbf{f},\mathbf{n},\mathbf{b},\mathbf{t}\rangle$ that plays the role of classical negation.

Still on what regards congruentiality, \Cref{framePropertiesC} illustrates how a few variants of global contraposition happen to be validated by some of our logics.  We use there `ser' to refer to the class of serial frames, `sym' to refer to symmetric frames, and `any' to refer to arbitrary frames.

\begin{table}[hbt]
\begin{tabular}{l | c | l}
~ 
Rule & Negation
& Frame property \\\hline\hline
  $\frac{\varphi\;\models\;\psi}{\narb\psi\;\models\;\narb\varphi}$ & $\narb$ as either $\ineg$ or $\uneg$ & [any] \\[2mm]
  $\frac{\narb\varphi\;\models\;\psi}{\narb\psi\;\models\;\varphi}$ & $\narb$ as $\ineg$ & [sym] \\[2mm]
  $\frac{\varphi\;\models\;\narb\psi}{\psi\;\models\;\narb\varphi}$ & $\narb$ as $\uneg$ & [sym] \\[2mm]
  $\frac{\ineg\varphi\;\models\;\psi}{\uneg\psi\;\models\;\varphi}$ &  & [ref] \\[2mm]
  $\frac{\varphi\;\models\;\uneg\psi}{\psi\;\models\;\ineg\varphi}$ &  & [ref] \\[2mm]
  $\frac{\ineg\varphi\;\models\;\uneg\psi}{\psi\;\models\;\varphi}$ &  & [ref] or [ser+sym] \\[2mm]
\hline
\end{tabular}
\caption{}
\label{framePropertiesC}
\end{table}

\paragraph{Some of our ancestors}
Some terminological conventions and some concepts used in the present paper were borrowed or adapted from other fonts, sometimes without the due pause for inserting an explicit reference.  For instance, in Section~\ref{sec:intro}, dadaistic and nihilistic models come from~\cite{jmarcos:neNMLiP}, and that paper also introduces the connectives~$\wsmile$ and~$\wfrown$ of the so-called Logics of Formal Inconsistency (cf.~\cite{car:jmar:Taxonomy}) and the dual Logics of Formal Undeterminedness (cf.~\cite{jmarcos:neNMLiP}, where the adjustment connectives are called connectives `of perfection').  The minimal conditions on negation, called $\llbracket$\textit{falsificatio}$\rrbracket$ and $\llbracket$\textit{verificatio}$\rrbracket$, come from~\cite{Mar:OnPLR}.
The `strengthened models' from Section~\ref{sec:proofsystem} correspond to models with strongly-legal valuations in the terminology of~\cite{lah:avr:Unified2013}.
%
%
Most rules in Sections~\ref{sec:proofsystem} and \ref{sec:extensions} may be seen as negative counterparts of the corresponding rules found at \cite[Ch.~3]{fit:proofmethods}, \cite{Wansing}, \cite[Ch.~2]{poggiolesi2010gentzen} and \cite[Ch.~6]{ind:natded-etc}.
The rule for \GPKD, for instance, may be thought of as a variation on the following well-known sequent rule for the modal logic $KD$: 
$\frac{\g\Ra}{\Box\g\Ra}$.  
Also, the rule for \GPKF\ is a variation on the sequent rule from~\cite{kawai1987sequential} for the `Next' operator in the temporal logic $LTL$, namely: 
$\frac{\g\Ra\d}{\Box\g\Ra\Box\d}$.
Furthermore, in Section~\ref{sec:analyticity}, the trick behind using three-valued models for addressing the admissibility of the cut rule goes at least as far back as \cite{Schutte60}.

\paragraph{Paraconsistency and paracompleteness}
Let~$\mostbasic$ denote the basic system obtained from~$\GPK$ by deleting rules $[\ineg{\Ra}]$ and $[{\Ra}\uneg]$. \Cref{allSystems} summarizes all systems investigated in the present paper, while \Cref{allBasicRules,allContextRelations} summarize the rules for $\ineg$ and $\uneg$, as well as their context relations.
 
\begin{table}[hbt]
\begin{tabular}{|c|c|}
\hline
$\pi_0$ 
& $\set{\tup{q_1\Ra\;;\;q_1\Ra},\tup{\Ra q_1\;;\;\Ra q_1}}$
\\\hline
$\pi_1$ 
& $\set{ \tup{q_1\Ra\;;\;\Ra\ineg q_1}, \tup{\Ra q_1\;;\;\uneg q_1\Ra}}$
\\\hline
$\pi_2$ 
& $\set{\tup{q_{1}\Ra\;;\;\Ra\ineg q_{1}},\tup{\Ra q_{1}\;;\;\ineg q_{1}\Ra}}$
\\\hline
$\pi_3$ 
& $\pi_{1}\cup\set{\tup{\ineg q_{1}\Ra\;;\;\Ra q_{1}},\tup{\Ra\uneg q_{1}\;;\; q_{1}\Ra}}$
\\\hline
$\pi_4$ 
& $\pi_{1}\cup\set{\tup{\uneg q_{1}\Ra\;;\;\uneg q_{1}\Ra},\tup{\Ra\ineg q_{1}\;;\;\Ra\ineg q_{1}}}$
\\\hline
\end{tabular}
\caption{Context relations}
\label{allContextRelations}
\end{table}

\begin{table}[hbt]
\begingroup
\begin{tabular}{|c|c|c|}
\hline
~
& Sequent Rule
& Basic Rule
\\\hline
${[{\ineg}{\Ra}]}$ 
& $\ssrul{\g\Ra\varphi, \d}{\uneg\d,\ineg\varphi\Ra \ineg\g} $
& $\tup{\Ra p_1;\pi_1}\rs\ineg p_1\Ra$
\\\hline
${[{\Ra}{\uneg}]} $
& $\ssrul{\g,\varphi\Ra \d}{\uneg\d\Ra \uneg\varphi,\ineg\g}$
& $\tup{p_1\Ra ;\pi_1}\rs\Ra\uneg p_1$
 \\\hline
$[\D]$ 
& $\ssrul{\Gamma\Ra\Delta}{\uneg\Delta\Ra\ineg\Gamma}$
& $\tup{\Ra\;;\;\pi_{1}}\rs\Ra$
 \\\hline
$[\T_{1}]$
& $\ssrul{\Gamma,\varphi\Ra\Delta}{\Gamma\Ra\ineg\varphi,\Delta}$
& $\tup{p_{1}\Ra\;;\;\pi_{0}}\rs\Ra\ineg p_{1}$
\\\hline
 $[\T_{2}]$
 & $\ssrul{\Gamma\Ra\varphi,\Delta}{\Gamma,\uneg\varphi\Ra\Delta}$
& $\tup{\Ra p_{1}\;;\;\pi_{0}}\rs\uneg p_{1}\Ra\;$
\\\hline
$[\FUNC]$ 
&$\ssrul{\Gamma\Ra\Delta}{\ineg\Delta\Ra\ineg\Gamma}$
& $\tup{\Ra\;;\;\pi_{2}}\rs\Ra$
\\\hline
$[\B_{1}]$ 
& $\ssrul{\Gamma,\ineg\Gamma',\varphi\Ra\Delta,\uneg\Delta'}{\uneg\Delta,\Delta'\Ra\uneg\varphi,\ineg\Gamma,\Gamma'}$
& $\tup{ p_{1}\Ra\;;\;\pi_{3}}\rs\Ra\uneg p_{1}$
\\\hline
 $[\B_{2}]$ 
 & $\ssrul{\Gamma,\ineg\Gamma'\Ra\varphi,\Delta,\uneg\Delta'}
 {\uneg\Delta,\Delta',\ineg\varphi\Ra\ineg\Gamma,\Gamma'}$
 & $\tup{\Ra p_{1}\;;\;\pi_{3}}\rs\ineg p_{1}\Ra$
 \\\hline
$[\KFour_{1}]$
& $\ssrul{\uneg\Gamma,\Gamma',\varphi\Ra\ineg\Delta,\Delta'}{\uneg\Gamma,\uneg\Delta'\Ra\uneg\varphi,\ineg\Delta,\ineg\Gamma'}$
& $\tup{ p_{1}\Ra\;;\;\pi_{4}}\rs\Ra\uneg p_{1}$
\\\hline
$[\KFour_{2}]$
& $\ssrul{\uneg\Gamma,\Gamma'\Ra\varphi,\ineg\Delta,\Delta'}{\uneg\Gamma,\uneg\Delta',\ineg\varphi\Ra\ineg\Delta,\ineg\Gamma'}$
& $\tup{\Ra p_{1}\;;\;\pi_{4}}\rs\ineg p_{1}\Ra$
\\\hline
$[\D_{B}]$
& $\ssrul{\ineg\Gamma',\Gamma\Ra\Delta,\uneg\Delta'}{\Delta',\uneg\Delta\Ra\ineg\Gamma,\Gamma'}$
& $\tup{ \Ra\;;\;\pi_{3}}\rs\Ra$
\\\hline
$[\D_{4}]$
& $\ssrul{\uneg\Gamma',\Gamma\Ra\Delta,\ineg\Delta'}{\uneg\Gamma',\uneg\Delta\Ra\ineg\Gamma,\ineg\Delta'}$
& $\tup{ \Ra\;;\;\pi_{4}}\rs\Ra$ 
\\\hline
\end{tabular}
\caption{Sequent rules}
\label{allBasicRules}
\endgroup
\end{table}

\begin{table}[hbt]
\begin{tabular}{|c|c|}
\hline
$\GPK$
& $\mostbasic+[\uneg\Ra]+[\ineg\Ra]$
\\\hline
$\GPKD$
& $\GPK+[\D]$
\\\hline
$\GPKT$
& $\GPK+[\T_{1}]+[\T_{2}]$
\\\hline
$\GPKF$
& $\mostbasic+[\FUNC]+[\ineg=\uneg]$
\\\hline
$\GPKB$
& $\mostbasic+[\B_{1}]+[\B_{2}]$
\\\hline
$\GPKFour$
& $\mostbasic+[\KFour_{1}]+[\KFour_{2}]$
\\\hline
$\GPKDB$
& $\GPKB+[\D_{\B}]$
\\\hline
$\GPKFourD$
& $\GPKFour+[\D_{\KFour}]$
\\\hline
\end{tabular}
\caption{Sequent systems}
\label{allSystems}
\end{table}

The first column of \Cref{frameProperties} recalls some of the basic consecutions characteristic of classical negation and its interaction with $\land$, $\lor$, $\top$ and $\bot$, identifying in the second and third columns the conditions on frames that suffice to validate them.  
In this table, [ser], [ref], [sym], [trn] and [fun] refer, respectively, to serial, reflexive, symmetric, transitive and functional frames, and [any] refers to arbitrary frames.

\begin{table}[hbt]
\scriptsize
\begin{tabular}{l | c | c | l}
~ 
Consecution & $\narb$ as $\ineg$ 
& $\narb$ as $\uneg$
& Derivability adjustments \\\hline\hline
  $\narb\top\models$ &  [any] &  [ser]  &  \\
  $\models\narb\bot$ &  [ser] &  [any] & \\[1mm]
\hline
  $\varphi,\narb\varphi\models$ 
    & 
    & [ref]
    & $\wsmile\varphi,\varphi,\ineg\varphi\models$ [any] \\
  $\models\narb\varphi,\varphi$ 
    & [ref]
    &  
    & $\models\uneg\varphi,\varphi,\wfrown\varphi$ [any]\\[1mm]
\hline
  $\narb\narb\varphi\models\varphi$ 
    & [sym]
    &  
    & $\uneg\uneg\varphi\models\varphi,\wfrown\varphi$ [ser+trn] \\
  $\varphi\models\narb\narb\varphi$ 
    &  
    &  [sym] 
    & $\wsmile\varphi,\varphi\models\ineg\ineg\varphi$ [ser+trn]\\[1mm]
\hline
  $\narb\varphi\lor\narb\psi \models \narb(\varphi\land\psi)$ 
    &  [any]
    &  [any] & \\
  $\narb\varphi\lor\psi \models \narb(\varphi\land\narb\psi)$ 
    &  
    & [sym] 
    & $\wsmile\psi,\ineg\varphi\lor\psi \models \ineg(\varphi\land\ineg\psi)$ [ser+trn]\\
  $\varphi\lor\narb\psi \models \narb(\narb\varphi\land\psi)$ 
    &  
    &  [sym] 
    & $\wsmile\varphi,\varphi\lor\ineg\psi \models \ineg(\ineg\varphi\land\psi)$ [ser+trn]\\
  $\varphi\lor\psi \models \narb(\narb\varphi\land\narb\psi)$ 
    &  
    &  [sym] 
    & $\wsmile\varphi,\wsmile\psi,\varphi\lor\psi \models \ineg(\ineg\varphi\land\ineg\psi)$ [ser+trn]\\[1mm]
  $\narb(\varphi\land\psi) \models \narb\varphi\lor\narb\psi$ 
  & [any]
  &  [fun] & \\
  $\narb(\varphi\land\narb\psi) \models \narb\varphi\lor\psi$ 
    & [sym]
    &  
    & $\uneg(\varphi\land\uneg\psi) \models \uneg\varphi\lor\psi,\wfrown\psi$ [trn] \\
  $\narb(\narb\varphi\land\psi) \models \varphi\lor\narb\psi$ 
    &  [sym]
    &  
    & $\uneg(\uneg\varphi\land\psi) \models \varphi\lor\uneg\psi,\wfrown\varphi$ [trn] \\
  $\narb(\narb\varphi\land\narb\psi) \models \varphi\lor\psi$ 
    &  [sym]
    & 
    & $\uneg(\uneg\varphi\land\uneg\psi) \models \varphi\lor\psi,\wfrown\varphi,\wfrown\psi$ [ser+trn] \\[1mm]
  $\narb\varphi\land\narb\psi \models \narb(\varphi\lor\psi)$ 
    &  [fun]
    &  [any] & \\
  $\narb\varphi\land\psi \models \narb(\varphi\lor\narb\psi)$ 
    &   
    &  [sym] 
    & $\wsmile\psi,\ineg\varphi\land\psi \models \ineg(\varphi\lor\ineg\psi)$ [trn]\\
  $\varphi\land\narb\psi \models \narb(\narb\varphi\lor\psi)$ 
    &   
    &  [sym] 
    & $\wsmile\varphi,\ineg\varphi\land\psi \models \ineg(\varphi\lor\ineg\psi)$ [trn]\\
  $\varphi\land\psi \models \narb(\narb\varphi\lor\narb\psi)$ 
    &   
    &  [sym] 
    & $\wsmile\varphi,\wsmile\psi,\ineg\varphi\land\psi \models \ineg(\varphi\lor\ineg\psi)$ [ser+trn]\\[1mm]
  $\narb(\varphi\lor\psi) \models \narb\varphi\land\narb\psi$ 
    &  [any]
    &  [any] & \\
  $\narb(\varphi\lor\narb\psi) \models \narb\varphi\land\psi$ 
    &  [sym]
    & 
    & $\uneg(\varphi\lor\uneg\psi) \models \uneg\varphi\land\psi,\wfrown\psi$ [ser+trn]\\
  $\narb(\narb\varphi\lor\psi) \models \varphi\land\narb\psi$ 
    &  [sym]
    & 
    & $\uneg(\uneg\varphi\lor\psi) \models \varphi\land\uneg\psi,\wfrown\varphi$ [ser+trn]\\
  $\narb(\narb\varphi\lor\narb\psi) \models \varphi\land\psi$ 
    &  [sym]
    & 
    & $\uneg(\varphi\lor\uneg\psi) \models \uneg\varphi\land\psi,\wfrown\varphi,\wfrown\psi$ [ser+trn]\\[1mm]
\hline
\end{tabular}
\caption{}
\label{frameProperties}
\end{table}

\normalsize

\paragraph{A richer language in which to study negative modalities}
Here is a meaningful illustration of the way in which \textbf{LFI}s and \textbf{LFU}s are said to `recover classical reasoning', by the addition of appropriate assumptions to the classical inferences whose validity has been lost by the move to a non-classical environment.
Let a negation~$\neg$ be added to positive classical logic (with~$\supset$ as the symbol for implication), and consider the standard form of \textit{reductio} according to which (conc)~$p$ follows from (prem) $(\neg p\supset q)\land(\neg p\supset\neg q)$.  Such a rule fails both when [CA] and when [DA] (namely, the `consistency assumption' and the `determinedness assumption' that are characteristic of classical negation, see \Cref{sec:goals}) are challenged.
As a matter of fact, when [CA] is not to be presumed, one might produce a counter-example to \textit{reductio} by finding a state of affairs satisfying both~$q$ and $\neg q$ while not satisfying~$p$, and when [DA] is not to be counted on, a state of affairs satisfying neither~$p$ nor $\neg p$ would provide a counter-example to \textit{reductio}.  This could be fixed if one replaced
(prem) for (prem$^\star$)  $\mathsf{C}q\land(\neg p\supset q)\land(\neg p\supset\neg q)$
and replaced (conc) for (conc$^\star$)  $p\lor\mathsf{D}p$,
adding thus a consistency assumption to the premise and a determinedness assumption to the conclusion.  It should be clear that (conc$^\star$) follows from (prem$^\star$).

The particular languages focused upon in the present paper, of course, do not include a primitive implication. The fourth column of \Cref{frameProperties} illustrates how some other important consecutions from classical logic may be recovered by some of the logics studied in this paper.  \Cref{framePropertiesB} recalls a few of the most important consecutions validated by \text{LFI}s introduced in the survey~\cite{car:jmar:Taxonomy} (namely, related to axioms dubbed (ci) and (ca)), and identify the conditions on frames that suffice to validate them.

\begin{table}[hbt]
\begin{tabular}{l | c | l}
~ 
Consecution & Property 
& Derivability adjustment \\\hline
\hline
  $\ineg\wsmile\varphi\models\varphi$ 
    & ---
    & $\ineg\wsmile\varphi\models\varphi,\wfrown\varphi$ [any]  \\
  $\ineg\wsmile\varphi\models\ineg\varphi$ 
    & [trn] 
    &  
    \\[1mm]
\hline
  $\wsmile\varphi,\wsmile\psi\models\wsmile(\varphi\land\psi)$ 
    & [any] 
    \\
  $\wsmile\varphi,\wsmile\psi\models\wsmile(\varphi\lor\psi)$ 
    & [any] 
    & 
    \\[1mm]
\hline
\end{tabular}
\caption{}
\label{framePropertiesB}
\end{table}

On what concerns the second pair of consecutions in the latter table (related to axiom (ca)), that deal with the `propagation of consistency', it is worth noticing, in the presence of a classical implication, that $\wsmile\varphi,\wsmile\psi\models\wsmile(\varphi\supset\psi)$ is \textit{not} a valid consecution in normal modal logics, but its variant $\wsmile\psi\models\wsmile(\varphi\supset\psi),\wfrown\varphi$ is validated by arbitrary frames.

\paragraph{On the availability of classical negation}
We have already pointed out in \Cref{sec:goals} how a classical negation might be defined within the basic normal modal logic of arbitrary frames with the help of the modal paraconsistent negation~$\ineg$ and an additional classical implication.  Now, within the modal logic of reflexive frames one may also define a classical negation with the help of implication and the modal paracomplete negation~$\uneg$, by simply setting ${\sim}\alpha:=\alpha\supset\uneg\alpha$.  Indeed, it is easy to see that the more general consecution $\models p, p\supset q$ would then be valid because of the meaning of \textit{classical} implication, and to check that $p, p\supset\uneg p\models$ is also valid one may use modus ponens and the reflexivity of the underlying frames.

Do the above observations still hold good if the classical implication~$\supset$ is replaced by the \textit{intuitionistic} implication?  To comment on that, we let $\to$ be some implication connective, and define ${\neg}_1\alpha:=\alpha\to\uneg\alpha$, ${\neg}_2\alpha:=\alpha\to\ineg(\alpha\to\alpha)$ and ${\neg}_3\alpha:=\alpha\to\uneg(\alpha\to\alpha)$, so as to briefly discuss in what follows the relations that involve classical negation and its non-classical modal cousins.  To be sure, some such relations have already been mentioned in previous sections and the corresponding consecutions are collected in \Cref{frameProperties,morenegations}.  The latter table also contains some double negation rules in which negations of different types interact. 

\begin{table}[hbt]
\begin{tabular}{l | c | c}
~ 
Consecution & $\narb$ as $\uneg$ 
& $\narb$ as $\ineg$
\\\hline\hline
  $\narb\varphi\models{\sim}\varphi$ & [ref] & --- \\
  ${\sim}\varphi\models\narb\varphi$ & --- & [ref] \\
\hline
\end{tabular}\bigskip\\ 

\begin{tabular}{l | c | c | l}
~ 
Consecution & \shortstack{$\narb_{1}$ as $\uneg$, \\$\narb_{2}$ as $\ineg$} & \shortstack{$\narb_{1}$ as $\ineg$, \\ $\narb_{2}$ as $\uneg$} & Derivability adjustment \\\hline\hline
$\narb_1\varphi\models\narb_2\varphi$ & [ser] & [fun] \\[1mm]
$\narb_{1}\narb_{2}\varphi\models\varphi$
& \begin{tabular}{c} [ref] or \\ {[ser+sym]}\end{tabular}
& 
&  $\ineg\uneg\varphi\models\varphi,\wfrown\varphi$ [trn] \\
$\varphi\models\narb_{1}\narb_{2}\varphi$
& 
& \begin{tabular}{c} [ref] or \\ {[ser+sym]}\end{tabular}
& $\wsmile\varphi,\varphi\models\uneg\ineg\varphi$ [trn] \\\hline
\end{tabular}
\caption{}
\label{morenegations}
\end{table}

To facilitate the discussion, we might hereupon say that a 1-ary connective~$\#$ is \textsl{contrary-forming} iff it is not $\#$-paraconsistent (that is, if it respects {$\llbracket$\#-explosion$\rrbracket$}), and $\#$ is \textsl{subcontrary-forming} iff it is not $\#$-paracomplete (that is, if it respects {$\llbracket$\#-implosion$\rrbracket$}).  Additionally, a \textsl{contradictory-forming} connective is said to be a 1-ary operator that is both contrary-forming and subcontrary-forming --- corresponding to what we here have called a `classical negation'.  
Generalizing the above, two formulas~$\varphi$ and~$\psi$ are called \textsl{contrary} according to a given logic if the consecution $\varphi,\psi\models$ holds good, and are called \textsl{subcontrary} if the consecution $\models\psi,\varphi$ holds good.

Note on the one hand that if~$\to$ is classical implication all the three negations defined above, $\neg_1$, $\neg_2$ and $\neg_3$, are subcontrary-forming connectives. 
On the other hand, even if we assume that~$\to$ is intuitionistic, we see that while $\neg_2$ is contrary-forming over the logic of arbitrary frames, to obtain the same effect with $\neg_3$ we need to consider the logic of serial frames, while for~$\neg_1$ reflexive frames are on demand.  

The nullary connective~$\bot$ (or any formula equivalent to it, such as $\ineg\top$ over the logic of arbitrary frames, or $\uneg\top$ over the logic of serial frames) may be said to be the \textit{strongest contrary} of any given formula (in the sense that it entails any other contrary of that very formula), or said to be a `global contrary' (as it works uniformly as the contrary to \textit{any} formula); something analogous may be said about~$\top$ as the strongest/global \textit{subcontrary} of any given formula.
Note, at any rate, that if $\#\varphi$ is equivalent to~$\bot$, for arbitrary~$\varphi$, then~$\#$ fails to be a negation (it does not respect the $\llbracket$\#-\textit{verificatio}$\rrbracket$ condition).
It is worth pointing out that if~$\to$ is intuitionistic implication then $\neg_2\alpha$ is the \textit{weakest contrary} to the formula~$\alpha$, over the logic of arbitrary frames, the intuitionistic-like $\neg_3\alpha$ plays the same role over the logic of serial frames, and $\neg_1\alpha$ plays that role over the logic of reflexive frames. Both~$\neg_2$ and~$\neg_3$ may be said then to produce a `local contrary' out of a global contrary. 
These are all issues discussed in~\cite{humb:contrsub:theoria:2005}, where the `basic logic of contrariety and subcontrariety' is expected to validate the following consecutions that witness the interaction between the two non-classical types of modal negations: (i) $\uneg\varphi\models\ineg\varphi$; (ii) $\uneg\ineg\varphi\models\varphi$; (iii) $\varphi\models\ineg\uneg\varphi$.  As we have seen in the present paper, irrespective of the presence of any kind of implication in the language, while item (i) is validated by the class of serial frames, to validate items (ii) and (iii) one should also impose the symmetry of the latter frames.  
All these facts can now be easily checked also using the corresponding sequent systems presented above.

Finally, as another warning concerning the extension of the logics studied in the present paper by the addition of implication, it is worth noting that in~\cite{humb:contrsub:theoria:2005} it is also shown that the identification between the contrary-forming and the subcontrary-forming negations (that happens, for instance, within the logic of total functional frames) causes a collapse of intuitionistic implication into classical implication.  As we have seen by the study of $PKF$ in the present paper, however, imposing the identification between~$\ineg$ and~$\uneg$ does \textit{not} mean that our non-classical negations end up collapsing into classical negation, if an implication is not available.
All that having been said, it is worth noting that the framework of basic systems~\cite{lah:avr:Unified2013} can be used to provide semantics for all systems studied in this paper augmented with a classical or an intuitionistic implication. Cut-admissibility, analyticity, and undefinability of classical negation in the resulting extensions are left as matter for future work.

\paragraph{What is to follow}
The main feature of our approach in the present paper has been to rely on theoretical technology built elsewhere and show how it may be adapted to the present study.  Our hope is that this should prove a beneficial methodology, and that the idea of obtaining completeness and cut-admissibility as particular applications of more general results will become more common, rather than proceeding always through \textit{ad hoc} completeness and cut-elimination theorems.


While we have directed our attention, in the present paper, to classes of frames that turned out to be particular significative from the viewpoint of the relation between negative modalities of different types, we envisage several very natural ways of extending this study.
%
A \textit{first} natural extension would be to look at other classes of frames that prove to be relevant from the viewpoint of sub-classical properties of negation.
For instance, 
it is easy to see that the class of frames with the Church-Rosser property validates $\ineg\ineg p\models \uneg\uneg p$, pinpointing an interesting consecution that involves the interaction between negations of different types.
Some other classes of frames deserving study do not seem to show the same amount of promise, from the viewpoint of paraconsistency or paracompleteness.  For instance, euclidean frames validate {$\llbracket\ineg$-explosion$\rrbracket$}
if in the set of formulas $\{\ineg p,p\}$ one replaces~$p$ by $\ineg r$ (this is explained by the fact that formulas of the form~$\ineg\varphi$ are consistent --- in other words, formulas of the form~$\wsmile\ineg\varphi$ turn out to be validated), and also validate {$\llbracket\uneg$-implosion$\rrbracket$} if in $\{\uneg p,p\}$ one replaces~$p$ by $\uneg r$; in contrast, the transitive frames studied in \Cref{transitivitySection} cause a similar behavior, but swapping the roles of~$\ineg r$ and~$\uneg r$ in replacing~$p$.  The latter phenomena (see the notion of `controllable explosion' in \cite{car:jmar:Taxonomy}) is summarized in \Cref{framePropertiesD}.

\begin{table}[hbt]
\begin{tabular}{l | c | c}
~ 
Consecution & $\narb$ as $\ineg$ 
& $\narb$ as $\uneg$
\\\hline\hline
  $\narb\varphi,\narb\narb\varphi\models$
    & [euc] & [ser+trn] 
    \\
  $\models\narb\narb\varphi,\narb\varphi$
    & [ser+trn] & [euc] 
    \\[2mm]
  $\narb\varphi,\ineg\narb\varphi\models$
    & --- & [trn] 
    \\
  $\models\uneg\narb\varphi,\narb\varphi$
    & [trn] & --- 
    \\[1mm]
\hline
    $\models\wnarb\narb\varphi$ &  [euc] &  --- 
    \\
    $\wnarb\narb\varphi\models$ &  --- &  [euc] 
    \\[2mm]
  %
    $\models\wsmile\narb\varphi$ &  --- &  [trn] 
    \\
    $\wfrown\narb\varphi\models$ &  [trn] &  --- 
    \\[1mm]
\hline
\end{tabular}
\caption{}
\label{framePropertiesD}
\end{table}

\noindent
A \textit{second} avenue worth exploring would lead us into logics containing more than one negative modality of the same type (as it has been done for logics with multiple paracomplete negations in~\cite{res:combpossneg:SL97}).
One could for instance consider not only the `forward-looking' negative modalities defined by the semantic clauses {[S$\ineg$]} and {[S$\uneg$]}, but also `backward-looking' negative modalities~$\ineg^{\!-1}$ and~$\uneg^{\!-1}$ defined by the clauses obtained from the latter ones by replacing $wRv$ by $vRw$ (such `converse modalities' have been studied in the context of temporal logic \cite{Prior67}, as well as in the context of the so-called Heyting-Brouwer logic \cite{rau:BH:DM1980}, and more recently have been given a treatment in terms of display logic calculi and multi-relational frames \cite{Onishi-sub15}).
The interaction between the various negations would then be witnessed, in such extended language, by the validity over arbitrary frames of `pure' consecutions such as $\ineg^{\!-1}\ineg p\models p$ and $\ineg\ineg^{\!-1} p\models p$ (as well as $p\models \uneg^{\!-1}\uneg p$ and $p\models \uneg\uneg^{\!-1} p$) and forms of global contraposition such as $\frac{\varphi\;\models\;\uneg\psi}{\psi\;\models\;\uneg^{\!-1} \varphi}$ and $\frac{\ineg\varphi\;\models\;\psi}{\ineg^{\!-1} \psi\;\models\;\varphi}$ (as well as $\frac{\varphi\;\models\;\uneg^{\!-1} \psi}{\psi\;\models\;\uneg\varphi}$ and $\frac{\ineg^{\!-1}\varphi\;\models\;\psi}{\ineg \psi\;\models\;\varphi}$), 
and by the validity over symmetric frames of `mixed' consecutions such as $\uneg^{\!-1}\ineg p\models p$ and $\uneg\ineg^{\!-1} p\models p$ (as well as $p\models \ineg^{\!-1}\uneg p$ and $p\models \ineg\uneg^{\!-1} p$).
In our view, it seems worth the effort applying the machinery employed in the present paper to the above mentioned systems, and still others, in order to investigate results analogous to the ones we have here looked at.
In particular, Proposition 4.28 of \cite{lah:avr:Unified2013}, that was used here for symmetric and transitive frames, provides a tool for constructing two accessibility relations that are inverse to one another, and thus one may serve as the `backward-looking' version of the other. %

Yet another direction for future research is the {\em combination} of certain frame properties. We have initiated such study in \Cref{serialityTransitivitySection}, where models that are both transitive and serial, or both symmetric and serial were addressed. 
The main challenge one would face pursuing this direction is that most basic rules discussed in the present paper employ different context relations, and hence, according to the semantics of basic systems, their combination characterizes models that are equipped with several distinct accessibility relations. In the case of \Cref{serialityTransitivitySection}, this was remedied by finding a different rule for seriality, that utilizes the same context-relation as the rule for transitivity. A similar solution was implemented for the combination of symmetry and seriality (captured in isolation by systems $\GPKB$ with $\GPKD$). To deal with such a combination, we had to adapt the method for extending partial models. Combining $\GPKD$ with either $\GPKT$ or $\GPKF$ is of course redundant, as reflexive relations and total functions are already serial.  Also,  the rules for $\GPKT$ employ $\pi_{0}$, which is already used for the classical connectives, and thus do not require an additional accessibility relation. Combining reflexivity with other frame properties seems therefore more promising in the context of cut-free sequent systems, however not if this is to involve other logics that are free of classical negation (given that already in $\GPKT$ classical negation is definable with the help of either one of the adjustment connectives).

One aspect that laid beyond the scope for the current paper was the prospect of using non-normal modal logics to define negative modalities. The idea is roughly the same: on the semantic side, interpreting the negative modal operators as 
impossibility and unnecessity,
and on the proof-theoretic side, converting derivation rules for non-normal modal logics into `negative' ones. 
From that perspective, it is worth noting that the `regular' and `co-regular' negations introduced in~\cite{vaka:cons89:full}, yet very natural, fail to be `full type' modalities.
Many proof systems for positive non-normal modal logics (e.g.\ all ordinary sequent systems from \cite{10.2307/20016216}) are actually basic systems, thus are amenable to an analysis similar to the one proposed here. We also note that positive and negative modalities within non-distributive logics are investigated in 
\cite{hartonas2016nondistributive} from an algebraic point of view, and in 
\cite{Greco07082016} the algebraic study of such logics is also complemented by proof-theoretical presentations in terms of display calculi.

At last, it is also worth stressing some specific problems that were left \textit{open} in the present paper.
For instance, an analytic proof system for negative modalities over euclidean frames (corresponding to the logic K5) was not presented here.  The sequent rule for K5 from \cite{takano2001modified} may indeed be adapted to the study of negative modalities; however, unlike what happens with the other cases studied here, the framework of basic systems falls short to handle this logic, as some of its strengthened models are not euclidean.
As another example, it was established in \Cref{sec:definability} that classical negation is definable in full $\GPKDB$, but not without~$\wsmile$ and~$\wfrown$. Whether only one of the latter adjustment connectives suffices for such a definition is left as matter for further research, as the methods developed here for attacking this question seem to leave it unanswered. Indeed, the definitions of classical negation for $\GPKT$ fail in $\GPKDB$, and the countermodels presented for $\wsmile$-free and $\wfrown$-free logics are not suitable for it. Finally, cut-admissibility of our system $\GPKB$ is currently left open. Note that for the situation in which the $\wsmile\wfrown$-fragment of $\GPKB$ is augmented with an implication connective, \cite{avr:zam:AiML16} includes an example of a derivable sequent that has no cut-free proof. However, the languages that we considered here, and in particular, the language of $\GPKB$, do not include a native implication connective.%
\footnote[7]{The authors acknowledge partial support by CNPq, by The Israel Science Foundation (grant no.\ 817-15), by the Marie Curie project GeTFun (PIRSES-GA-2012-318986) funded by EU-FP7, and by the Humboldt Foundation. 
They also take the chance to thank Elaine Pimentel, Heinrich Wansing, Alessandra Palmigiano, Dorota Leszczy\'nska-Jasion, and two anonymous referees for their comments on an earlier incarnation of this manuscript.
}